\documentclass[a4paper,UKenglish,cleveref, autoref, thm-restate, 
]{lipics-v2021}
\usepackage{graphicx} 

\newcommand{\tw}{\text{tw}}
\newcommand{\sn}{\text{sn}}
\newcommand{\eps}{\varepsilon}
\newcommand{\MST}{$\mathcal{EMST}$}

\usepackage{svg}
\usepackage{wrapfig}
\usepackage{float}

\nolinenumbers
\hideLIPIcs

\newcommand{\bigO}{\mathcal{O}}

\usepackage[noend]{algpseudocode}
\usepackage{algorithm}
\usepackage{thmtools} 

\usepackage{xpatch}
\usepackage{apptools}
\makeatletter
\xpatchcmd{\thmt@restatable}
{\csname #2\@xa\endcsname\ifx\@nx#1\@nx\else[{#1}]\fi}
{\IfAppendix{\csname #2\@xa\endcsname}{\csname #2\@xa\endcsname\ifx\@nx#1\@nx\else[{#1}]\fi}}
{}{} 
\makeatother

\algnewcommand\algorithmicforeach{\textbf{for each}}
\algdef{S}[FOR]{ForEach}[1]{\algorithmicforeach\ #1\ \algorithmicdo}

\title{Geometric spanners of bounded tree-width}

\author{Kevin Buchin}{TU Dortmund, Germany}{kevin.buchin@tu-dortmund.de}{https://orcid.org/0000-0002-3022-7877}{}
\author{Carolin Rehs}{TU Dortmund, Germany}{carolin.rehs@tu-dortmund.de}{https://orcid.org/0000-0002-8788-1028}{}
\author{Torben Scheele}{TU Dortmund, Germany}{torben.scheele@tu-dortmund.de}{https://orcid.org/0009-0006-6119-6598}{}

\authorrunning{K.~Buchin, C.~Rehs, T.~Scheele} 

\Copyright{Kevin Buchin, Carolin Rehs, Torben Scheele} 



\category{} 

\begin{document}

\maketitle
\begin{abstract}
Given a point set $P$ in the Euclidean space, a geometric $t$-spanner $G$ is a graph on $P$ such that for every pair of points, the shortest path in $G$ between those points is at most a factor $t$ longer than the Euclidean distance between those points. The value $t\geq 1$ is called the dilation of $G$. 
Commonly, the aim is to construct a $t$-spanner with additional desirable properties.
In graph theory, a powerful tool to admit efficient algorithms is bounded tree-width.
We therefore investigate the problem of computing geometric spanners with bounded tree-width and small dilation $t$. 

Let $d$ be a fixed integer and $P \subset \mathbb{R}^d$ be a point set with $n$ points. We give a first algorithm to compute an $\mathcal{O}(n/k^{d/(d-1)})$-spanner on $P$  with tree-width at most $k$. The dilation obtained by the algorithm is asymptotically worst-case optimal for graphs with tree-width $k$: We show that there is a set of $n$ points such that every spanner of tree-width $k$ has dilation $\mathcal{O}(n/k^{d/(d-1)})$.
We further prove a tight dependency between tree-width and the number of edges in sparse connected planar graphs, which admits, for point sets in $\mathbb{R}^2$, a plane spanner with tree-width at most $k$ and small maximum vertex degree. 

Finally, we show  an almost tight bound on the minimum dilation of a spanning tree of $n$ equally spaced points on a circle, answering an open question asked in previous work.
\end{abstract}

\section{Introduction}

Geometric spanners are an extensively studied area in computational geometry, see \cite{Bose.2013,Narasimhan.2007} for surveys.
 A geometric spanner for a point set $P \in \mathbb{R}^d$ is a weighted graph on $P$, where the weight of an edge is the Euclidean distance between its endpoints, aiming for sparsity while avoiding long paths between two points in $P$. This is formalised by the dilation $t= \max\left\{ \frac{d(p,p')}{|pp'|} \mid p,p' \in P \right\}$, where $d(p,p')$ is the distance between two points in $G$ and $|pp'|$ is the Euclidean distance between $p$ and $p'$. A geometric spanner with dilation at most $t$ is also called a \emph{$t$-spanner}.
 
When geometric spanners were introduced in 1986~\cite{chew1986}, sparsity was obtained by requiring planarity. Since then, plane geometric spanners have been widely studied~\cite{Bose.2013}. Plane geometric spanners with small dilation can be obtained by different triangulations. An important tool here is the Delaunay triangulation which has a dilation of at most 1.998~\cite{xia2013}. Further triangulations that lead to a constant stretch factor are the greedy triangulation and the minimum weight triangulation, which both fulfil the so-called $\alpha$-diamond-property~\cite{DasJ89,drysdale2001exclusion}. By~\cite{DasJ89}, every triangulation that fulfils this property is a geometric spanner with dilation depending on $\alpha$.

If, instead of minimising $t$ for a fixed number of edges, we ask for a graph with at most some fixed dilation $t$ and a minimal number of edges, the greedy spanner, which inserts an edge $\{u,v\}$ whenever there is not already a path of length at most $t$ between $u$ and $v$, achieves asymptotically optimal edge number. It can be computed in $\mathcal{O}(n^2 \log n)$ time using $\mathcal{O}(n^2)$ space~\cite{bose2010computing}. Further, the greedy spanner has asymptotically optimal edge count, and bounded maximum vertex degree.

But also other measures are of interest: Especially spanners with a linearly bounded number of edges or a bounded vertex degree have been widely researched~\cite{klein2007computing,Narasimhan.2007,smid2007well,gudmundsson2018dilation,bose2010computing,das1996,kanj.2017,bonichon2010}.

In graph theory, an important tool to design efficient algorithms for in general NP-hard problems is to regard these problems on graphs which are bounded in certain parameters. Especially graphs with bounded \emph{tree-width}~\cite{RS83, RS86} allow polynomial-time algorithms for many in general NP-hard problems using a dynamic programming algorithm on the tree decomposition. 
By Courcelle's Theorem, every problem that is describable in monadic second order logic with vertex and edge quantification is solvable in polynomial time for graphs with bounded tree-width~\cite{CO00}.
Thus, geometric spanners with bounded tree-width would allow polynomial time solutions for many in general NP-hard problems and thus be a strong tool for many geometric applications. For instance, Cabello and Knauer~\cite{CABELLO2009815} recently gave some examples for efficient algorithms for geometric problems on graphs with bounded tree-width. They show that via orthogonal range searching, for graphs of $n$ vertices and tree-width $k$ with $k\geq 3$, the sum of the distances between all pairs of vertices can be computed in $\bigO(n \log^{k-1} n)$ time. Further, they show that the dilation of a geometric graph of bounded tree-width can be computed in $\bigO(n \log^{k+1} n)$ time. That makes it possible to compute the dilation of every tree-width bounded graph, but does not lead to tree-width bounded spanners with low dilation. 
In~\cite{Cabello2022Computing}, Cabello shows how to compute the inverse geodesic length in graphs of bounded tree-width. Further, Cabello and Rote~\cite{cabello2010obnoxious} show how to compute obnoxious centers in graphs with bounded tree-width in $\bigO(n \log n)$ time.

Therefore, we aim to construct spanners with bounded tree-width. There are spanners with sublinear tree-width. More specifically, since planar graphs have tree-width $\bigO(\sqrt{n})$ \cite{Lipton1979} (where $n$ is the number of vertices), for point sets in the Euclidean plane, constant dilation spanners with tree-width $\bigO(\sqrt{n})$ can be obtained by plane spanner approaches such as the Delaunay triangulation, or the greedy triangulation. For point sets in $\mathbb{R}^d$, we point out that the greedy spanner has tree-width $\bigO(n^{1-1/d})$, using separator results from \cite{Le2024}. 

While this gives some first results, constructing spanners with tree-width $k \in o(n^{1-1/d})$ will require other constructions. The only known construction for small tree-width and a guarantee on the dilation is the Euclidean minimum spanning tree: it has tree-width $k=1$ and worst-case dilation $n-1$. In~\cite{klein2007computing,CHEONG2008188} the authors show that it is NP-hard to compute a minimal dilation tree. Since trees are graphs with tree-width $1$, computing a minimal dilation spanner with bounded tree-width is also NP-hard.

We present a first algorithm that provides a trade-off between tree-width and the dilation:

\begin{restatable}{theorem}{upperBoundHigherDim}
\label{thm:upper bound higher dim}
Given a set $P\subset \mathbb{R}^d$ of $n$ points for some fixed $d$ and a positive integer $k\leq n^{1-1/d}$. There is a geometric spanner of tree-width $k$ on $P$ with a dilation of $\bigO(n/k^{d/(d-1)})$ and bounded degree that can be computed in time $\bigO(n^2\log n)$.
\end{restatable}

We complement our result with a matching lower bound. 

\begin{restatable}{theorem}{lowerHigherDim}
\label{thm:lower higher dim}
Let $d\geq 2$ be a fixed integer. For positive integers $n$ and $k \leq {n^{(d-1)/d}\cdot (5d)^{1/d-2}}$, there is a set of $n$ points in $\mathbb{R}^d$, so that every geometric spanner of tree-width $k$ on this set has dilation $\Omega (n/k^{d/(d-1)})$.
\end{restatable}

This proves that our algorithm for \cref{thm:upper bound higher dim} has a worst-case optimal trade off between tree-width and dilation. We therefore now turn our attention to bounded tree-width spanners with additional property. The most commonly studied property of geometric spanners is planarity.

For point sets in the Euclidean plane, we can adapt the construction of the spanner for \cref{thm:upper bound higher dim} to obtain a planar spanner with bounded tree-width (and bounded degree), that is not necessarily plane. 
As a general tool to obtain  bounded tree-width plane spanners, we provide the following strong dependency between tree-width and the number of edges in sparse connected planar graphs:

 \begin{restatable}{theorem}{twtreeaddedges}
 \label{thm:twtreeaddedges}
 Let $G$ be a planar connected graph with $n$ vertices and $n+\lfloor (k-1)^2/72 \rfloor$ edges, where $k\geq 2$ is an integer. The tree-width of $G$ is at most $k$. 
 \end{restatable}

 This result is quite surprising, since for non-planar graphs, arbitrarily placed edges increase the tree-width linearly in the number of added edges, as there are graphs with $m$ edges and tree-width at least $m\cdot \eps$ for a constant $\eps$ \cite{Grohe2009}. 

 It further yields an alternative proof for \cref{thm:upper bound higher dim} for $d=2$, creating a bounded tree-width spanner with additional properties: 

 \begin{restatable}{corollary}{corPlaneTreewidthSpanner}
\label{cor:PlaneTreewidthSpanner}
Given a set $P\subset \mathbb{R}^2$ of $n$ points and some positive integer $k \leq 12 \sqrt{n-3}$, a plane spanner with tree-width $k$, maximum vertex degree $4$ and dilation $\bigO(n/k^2)$ can be constructed in $\bigO(n \log n)$ time. 
 \end{restatable}

This follows immediately using an adapted version of the algorithm provided in \cite{Aronov2008}: Instead of the Delaunay triangulation, use a plane constant dilation spanner with maximum vertex degree $4$, as provided in \cite{kanj.2017} and the MST of this spanner instead of the EMST.

In the last section, we answer an open question stated in~\cite{Aronov2008} by giving some smaller results for spanners of bounded tree-width on points in convex position. In particular, we show that, for point sets $P$ which are equally spaced on a circle, there is a tree spanner with dilation $\leq \frac{2n}{\pi}+\frac{\pi}{2n}$. This complements a lower bound given in~\cite{Aronov2008}.
Further, since plane spanners on points in convex position are outerplanar, there is a $1.88$-spanner of tree-width $2$ for $P$. Further, we give a (straight-forward) lower bound: There are sets of points in convex position on which every spanner of tree-width $2$ has dilation at least $1.43$.

\section{Preliminaries}

We start giving some preliminary definitions to clarify the notations we use later on. 

In this paper, all graphs $G$ are simple and undirected. We refer to the vertex set of $G$ by $V(G)$ and to the edge set of $G$ by $E(G)$. Euclidean graphs are complete and weighted and the edge weight is given by the Euclidean distance of the incident vertices. 
We denote the \emph{degree} of a vertex $v \in V(G)$ by $\Delta_G(v)$ and the \emph{maximum vertex degree} of a graph by $\Delta(G)$. By $K_n$ we denote the complete graph on $n$ vertices.

A \emph{Euclidean $2$-dimensional $n \times n$-grid} is defined as a graph on the point set $\{0, \dots , n-1\}\times \{0, \dots ,n-1\}$ in $\mathbb{R}^2$. We call two vertices $(x,y)$ and $(x',y')$ \emph{neighbouring}, if $\vert x-x'\vert + \vert y-y'\vert = 1$. The edge set of the $n \times n$ grid is then $\{\{(x,y),(x',y')\} \mid (x,y) \text{ and } (x',y') \text{ are neighbouring}\}$.

A \emph{Euclidean $d$-dimensional $n^d$-grid} is defined as a graph on the point set $\{0, \dots , n-1\}^d$. Here, two vertices $(x_1,\dots,x_d)$ and $(x_1',\dots ,x_d')$ are called \textit{neighbouring} if $\vert x_1-x_1'\vert +\dots  +\vert x_d-x_d'\vert = 1$. The edge set of the $n^d$ grid is then  $\{\{(x_1,\dots x_d),(x_1',\dots,x_d')\} \mid (x_1,\dots x_d) \text{ and } (x_1',\dots x_d') \text{ are neighbouring}\}$.

We use a standard definition of tree-width~\cite{RS86} in the following notation: 
A \emph{tree decomposition} $(T,X)$ is a tree $T$ over a set of nodes $I$ and a mapping $X: I \rightarrow 2^V$, where $X(i)$ is called the \emph{bag} of $i$, so that the following conditions hold:
\begin{enumerate}
\item $\bigcup_{i\in I}X(i)=V$
\item $\bigcup_{i\in I}E(i)=E$ where $E(i)=\{\{u,v\}\in E \mid u,v \in X(i)\}$
\item If $v\in X(i)$ and $v\in X(k)$ for some $i,k\in I$ then $v\in X(j)$ holds for every $j$ on the path in $T$ between $i$ and $k$.
\end{enumerate} 
The \textit{width} of a tree decomposition is defined as $\max_{ i\in I} \vert X(i)\vert-1$. The \textit{tree-width} $\tw(G)$ of $G$ is the minimum width over all tree decompositions of $G$.

A minor of a graph can be obtained by subgraph operations and edge contractions. As, however, edge contractions in general do not maintain any properties of the embedding, which is essential for geometric graphs, we give the definition of contracting one vertex to another: 
For a graph $G$, \textit{contracting vertex} $u$ to $v$ or equivalently \textit{contracting vertex} $u$ along $e$, we create a minor $G'=(V',E')$ of $G$ with $V'=V\setminus \{u\}$ and $E'=\{\{u',v'\} \mid u',v'\in V', \{u',v'\} \in E\} \cup \{\{u',v\}\mid u'\neq v, \{u',u\}\in E\}$. Note that this is just another way of defining an edge contraction. 
Given a graph $G=(V,E)$, some vertex $v\in V$ and a path $\gamma=v,v_2,...,v_\ell$ inside $G$, we define the \emph{path contraction} of $\gamma$ to $v$ to be the series of vertex contractions starting with the contraction of $v_2$ to $v$ up to the contraction of $v_\ell$ to $v$.
A graph $H$ is a \emph{minor} of $G$ if it can be obtained by vertex contractions and subgraph operations. A class of graphs $\mathcal{G}$ is considered \emph{minor closed} if every minor $H$ of a graph $G\in \mathcal{G}$ is also in $\mathcal{G}$.

A \emph{$t$-balanced separator} of size $s$ of a graph $G$ is an edge set $S\subseteq V(G)$, $|S|=s$ such that $V-S$ can be partitioned into two sets $A$ and $B$ with 
no edges between the vertices of $A$ and $B$ and $\vert A\vert\leq t\cdot \vert V\vert$ as well as $\vert B\vert \leq t\cdot \vert V\vert$. If $t=\frac{2}{3}$, we call $S$ a \emph{separator} of $G$.

\section{Bounded tree-width spanners in higher dimensions}\label{sec:general}

In this section, we present an algorithm for constructing bounded tree-width spanners for points in $\mathbb{R}^d$ for constant dimension $d$. As we will see later on, the trade-off between tree-width and dilation of this spanner is worst-case optimal. 

The algorithm makes use of two ingredients: a Euclidean minimum spanning tree (\MST) and a class of spanners with sublinear tree-width. For $d=2$ a class of planar spanners can be used (and will be used in Section~\ref{sec:planar}). For general $d$ we can use the greedy spanner~\cite{Althfer1993} instead.

To bound the tree-width of the greedy spanner, we make use of separators.
The \textit{separating-number} $\sn(G)$ of a graph $G$ as the smallest integer $s$ such that every subgraph of $G$ has a separator of size $s$. 

\begin{lemma}[\cite{Dvok2019}]
\label{lem:separatingnumber}
The tree-width of any graph $G$ is at most $15\sn(G)$.
\end{lemma}

Recently, Le and Than~\cite{Le2024} proved that subgraphs of the greedy spanner have small separators. This generalizes a corresponding result in two dimensions~\cite{EppsteinK21}.
\begin{lemma}[\cite{Le2024}]
\label{lem:greedysep} 
Let $P\subset \mathbb{R}^d$ be a set of points and $G$ be the greedy $t$-spanner of $P$ for $t\in (1,3/2]$. Every subgraph of $G$ on $k$ vertices has a $(1-\frac{1}{\eta_d2^{d+1}})$-balanced separator of size $c_{d} (t-1)^{1-2d} k^{1-1/d}$, where $c_{d}\leq 2^{\bigO(d)}$ is a constant and $\eta_d$ is the packing constant of $d$-dimensional Euclidean space.
\end{lemma}

The packing constant of $d$-dimensional Euclidean space is defined as the smallest number $\eta_d$ such that for any $r\in (0,1]$ and any set of points $P\in \mathbb{R}^d$ contained in the $d$-dimensional unit ball, where the minimum distance between any pair of points of $P$ is $r$,  we have that $\vert P\vert \leq \eta_d r^{-d}$. The packing constant can be bounded by $\eta_d = 2^{\bigO(d)}$. We refer to the paper of Le and Than \cite{Le2024} for more details.

Using the fact that subgraphs of the greedy spanner have small separators, we can derive a bound on the tree-width of the greedy spanner. We fix $t=3/2$ in order to obtain a bound independent of $t$.

\begin{lemma}
\label{lem:greedytw}
Let $P\subset \mathbb{R}^d$ be a set of $n$ points and $G$ the greedy $3/2$-spanner of $P$. $G$ has tree-width at most $15 \eta_d c_d 8^d \cdot n^{1-1/d}$, where
$\eta_d$ is the packing constant in $d$-dimensional Euclidean space and $c_{d}\leq 2^{\bigO(d)}$ is a constant.
\end{lemma}
\begin{proof}
To bound the tree-width of $G$, we show that every subgraph of $G$ has a separator of size $\eta_d c_d 8^d \cdot n^{1-1/d}$ and therefore bound the separating-number of $G$. Let $H$ be a subgraph of $G$ on $k$ vertices. By~\cref{lem:greedysep} there is a $(1-\frac{1}{\eta_d2^{d+1}})$-balanced separator of size $c_{d} 2^{2d-1} k^{1-1/d}$ for $H$. To obtain a separator for $H$ we can iteratively combine separators for the currently largest remaining connected component, induced by the current separator, until the largest remaining component is at most of size $2/3\cdot k$. It is not hard to see that after $i$ iterations the largest induced connected component is of size at most $\left( 1-\frac{1}{\eta_d2^{d+1}}\right)^i\cdot k$. We claim that after $\eta_d2^{d+1}$ iterations the largest remaining induced component is of size $\leq 2/3\cdot k$. Since $\eta_d2^{d+1} \geq 2$, we have that
$$\left(1-\frac{1}{\eta_d2^{d+1}} \right)^{\eta_d2^{d+1}}\leq e^{-1} \leq 2/3,$$
which proves the claim. The size of the resulting separator can be bounded by
\begin{align*}
\sum_{j=0}^{\eta_d2^{d+1}}  c_{d} 2^{2d-1} \cdot \left(k\cdot \left(1-\frac{1}{\eta_d2^{d+1}}\right)^j\right)^{1-1/d}\\
\leq \eta_d  c_d  2^{d+1}  2^{2d-1} \cdot k^{1-1/d} = \eta_d c_d 8^d \cdot k^{1-1/d}.
\end{align*}
Since every subgraph of $G$ has a separator of size at most $\eta_d c_d 8^d \cdot n^{1-1/d}$, we can bound the separating-number of $G$ by $\eta_d c_d 8^d \cdot n^{1-1/d}$ and consequently by~\cref{lem:separatingnumber} we can conclude that the tree-width of $G$ is $\leq 15\cdot \eta_d c_d 8^d \cdot n^{1-1/d}$.
\end{proof}

The following algorithm computes a bounded tree-width spanner. The constant $C$ is defined as $C=30\eta_dc_d8^d$.

\begin{algorithm}[H] \caption{$d\texttt{-DimensionalBoundedTreeWidthSpanner}(P,k)$}
\label{alg:higherdim}
\begin{algorithmic}[1]
\Require a set of $n$ points $P\subset \mathbb{R}^d$ and a natural number $k\leq n^{1-1/d}$
\Ensure an $\bigO(n\cdot k^{(d-1)/d})$-spanner $G=(P,E)$ of tree-width $k$ 
\State $\mathcal{EMST} \leftarrow $ Euclidean minimum spanning tree of $P$
\If {$k=1$}
  \State { \Return $\mathcal{EMST}$}
\EndIf
\State $m \leftarrow \left\lceil \left(k/C\right)^{d/(d-1)} +1\right\rceil$
\State Compute a set $\mathcal{T}$ of $m$ disjoint subtrees of $\mathcal{EMST}$, each containing $\mathcal{O}(n/m)$ points, as explained in the proof of \cref{lem:subtrees}.
\State For $T \in \mathcal{T}$ let $R(T)$ be the vertices in $T$ incident to edges removed by the previous step.
\State For each $T \in \mathcal{T}$ iterate through its edges $e$ from long to short. If $e$ lies on a path between two vertices in $R(T)$, remove the edge. Let $E'(T)$ be the set of remaining edges.
\State $E'\leftarrow \bigcup_{T\in\mathcal{T}}E'(T)$
\State $(R, E'') \leftarrow$ greedy $3/2$-spanner for $R = \bigcup_{T\in\mathcal{T}}R(T)$
\State \Return $G=(P,E'\cup E'')$
\end{algorithmic}
\end{algorithm}
We now start with a detailed description on how the subtrees (line 5) are constructed. 
The following is a folklore result and is for instance mentioned in~\cite{Lipton1979}.

\begin{lemma}
\label{lem:vertsep}
For any tree $T$ on $n$ vertices there is a separator vertex that splits $T$ into subtrees each at most of size $n/2$ and can be found in time $\bigO(n)$.
\end{lemma}

\begin{proof}
For the sake of a contradiction, suppose there is no vertex in $T$ that splits $T$ into subtrees each at most of size $n/2$. Now let $v$ be a vertex that minimizes the maximum size of the subtrees induced by its removal. We know that there must be a subtree $T'$ so that $\vert T' \vert > n/2$. The sum of the sizes of all the other subtrees $T_1,...,T_k$ induced by the removal of $v$ therefore must be less than $n/2-1$. Let $u$ be the vertex of $T'$ that is connected to $v$. Removing $u$ instead of $v$ leaves us with one subtree of size less than $n/2$, this subtree being the union of $T_1,...,T_k$ and $v$ and possibly several other subtrees, which were a part of $T'$ before and are therefore now each of size strictly smaller than the size of $T'$. This of course contradicts with $v$ being a vertex that minimizes the maximum size of the subtrees induced by its removal. We can therefore conclude that there must be a separator vertex that splits $T$ into subtrees each at most of size $n/2$.

The argument above even lets us deduce an algorithm for finding the separator vertex in question. Start by arbitrarily choosing a root vertex. At this root begin a DFS and recursively compute the sizes of the subtrees rooted at the children of the current vertex. With $s(v)$ we denote the size of the subtree rooted at $v$. If the current vertex $v$ is a leaf we set $s(v)=1$. For any other vertex $s(v)$ is the sum of the sizes of the subtrees rooted at the children of $v$ plus one. Now go through every vertex $v$ and check if for every child $u$ it holds that $s(u)\leq n/2$ and that $n-s(v)\leq n/2$. If $v$ is the separator vertex in question, we know that every subtree induced by the removal of $v$ has size at most $n/2$ and therefore for every child $u$ of $v$ it holds that $s(u)\leq n/2$ and $n-s(v)\leq n/2$. Since $T$ is a tree, the DFS takes time $\bigO(n)$ and since we have to check every vertex only once to find a separator vertex, we end up with a running time of $\bigO (n)$.
\end{proof}

Considering the separator vertex $v$ and its edges, we observe that one of the by removal of $v$ induced subtrees must have size at least $n\cdot 1/\Delta(v)$. Let $e$ be the edge that connects $v$ to this subtree. Then, $e$ separates the tree into one subtree of size at least $n\cdot 1/\Delta(v)$, this being the aforementioned subtree, and one of size at most $n\cdot (\Delta(v) -1)/\Delta(v)$.

\begin{lemma}
\label{lem:edgesep}
For any tree $T$ of degree $\Delta$ with $n$ vertices there is a separator edge, that separates $T$ into two subtrees, one of size at least $n\cdot 1/(\Delta+1)$ and one of size at most $n\cdot \Delta/(\Delta+1)$. This edge can be found in time $\bigO(n)$.
\end{lemma}

This result can be extended as follows:

\begin{lemma}
\label{lem:subtrees}
Given a tree $T$ of degree $\Delta$, that has $n$ vertices and some integer $m\geq 1$, there is a set of $m-1$ edges of $T$, whose removal lead to $m$ disjoint subtrees of size $\bigO(n/m)$. 
This set of edges can be found in time $\bigO(n (\Delta+1)(\log m + \log(\Delta +1) ))$.
\end{lemma}
\begin{proof}
By Lemma \ref{lem:edgesep} for any tree of size $n$ there is always an edge, whose removal splits the tree into one subtree of size at most $n\cdot \Delta/(\Delta+1)$ and one of size at least $n\cdot 1/(\Delta+1)$. The idea now is to split subtrees until their sizes are sufficiently small. So given a tree on $n$ points and some $m\in\mathbb{N}$, we keep splitting the largest subtrees until there is no subtree of size more than $(\Delta +1)\cdot n/m$ left. Note that this procedure can only lead to subtrees as small as $n/m$ and therefore only to at most $m$ subtrees. To obtain exactly $m$ subtrees we can simply continue splitting the largest remaining subtrees until the number of subtrees is exactly $m$. The subtrees are now at least of size $n/((\Delta +1)\cdot m)$ and at most of size $(\Delta +1)\cdot n/m$ and therefore of size $\mathcal{O}(n/m)$.

Now let us look at the time needed to find this set. Let $\mathcal{T}$ be the set of every subtree that is created by one of the $m-1$ splits, including $T$. Every split of a subtree results in two new subtrees. We can therefore organize the subtrees of $\mathcal{T}$ into a binary tree where $T$ is the root and the two subtrees resulting from the first split are the children of $T$ and so on. It is not hard to observe that the subtrees of the same layer, that being the subtrees of the same depth in the binary tree, are pairwise disjoint and that every leaf in the binary tree is one of the $m$ subtrees that is returned by the above procedure. By Lemma \ref{lem:edgesep} we know that finding a separator edge inside a tree of size $n$ takes time $\bigO(n)$. The running time of the above procedure therefore is $\sum_{T'\in \mathcal{T}}\bigO(\vert T'\vert)$. The maximum number of splits necessary to obtain a subtree of size $n/((\Delta +1)\cdot m)$ is $(\Delta +1)\cdot (\log m + \log(\Delta+1))$. The binary tree therefore has height at most $h=\lceil (\Delta +1)\cdot (\log m + \log(\Delta+1)) \rceil$. Let $\mathcal{T}_i\subseteq \mathcal{T}$ be the set of subtrees in layer $i$ of the binary tree. Since the subtrees of the same layer are pairwise disjoint, we can conclude that for every $i\in \{0, \dots, h\}$ the union $\bigcup \mathcal{T}_i$ has size at most $n$. The running time of the procedure above can therefore be bounded by 

\[ \sum_{T'\in \mathcal{T}}\bigO(\vert T'\vert)=\sum_{i=0}^{h} \sum_{T' \in \mathcal{T}_i}\bigO\left(\vert T'\vert \right)=\sum_{i=0}^{h} \bigO\left( n \right)=\bigO\left( n h\right)=\bigO(n (\Delta+1)(\log m + \log(\Delta +1) )). \]
\end{proof}

We call a vertex a \emph{representative} of a subtree $T  \in \mathcal{T}$ if it was incident to an edge of the \MST\ that was removed in the construction of \cref{lem:subtrees}. 
While bounding the dilation of the spanner requires more work, we can already derive the running time and tree-width.

\begin{lemma}\label{lem:run}
Algorithm \ref{alg:higherdim} constructs in $\bigO(n^2\log n)$ time a graph with tree-width $\leq k$ and bounded degree.
\end{lemma}
\begin{proof}
The greedy spanners in line 9 of the algorithm can be computed in $\bigO(n^2\log n)$ time~\cite{bose2010computing}. 
The \MST\ can be computed in subquadratic time~\cite{AgarwalES91},
and all remaining steps take $\bigO(n\log n)$ time. The greedy $t$-spanner and the \MST\ both have bounded degree for fixed dimension (and $t$)~\cite{RobinsS95, Narasimhan.2007}.

The set of representatives of all subtrees has size $2m-2$, since each representative is incident to one of the $m-1$ edges removed. Thus, $\mathcal{GS}$ is the greedy spanner of $2m-2$ points. By Lemma~\ref{lem:greedytw} and the choice of the constant $C$ the tree-width of $\mathcal{GS}$ is at most $k$. After line 7 of the algorithm each of the subtrees of the subtrees $T \in \mathcal{T}$ contains only one representative. Thus, the graph $G$ consists of $\mathcal{GS}$ with a tree attached to each of its vertices. These trees do not increase the tree-width, since every tree contains only one vertex from $\mathcal{GS}$.
Thus the tree-width of $G$ is $k$.
\end{proof}

Next we bound the dilation of the computed spanner.

\begin{lemma}\label{lem:dilationd}
Given a set $P\subset \mathbb{R}^d$ of $n$ points for some fixed $d$ and a positive integer $k \leq n^{1-1/d}$, Algorithm \ref{alg:higherdim} computes an $\bigO(n/k^{d/(d-1)})$-spanner for $P$. 
\end{lemma}
\begin{proof}
For $k=1$, Algorithm \ref{alg:higherdim} returns the Euclidean minimum spanning tree of $P$, which is an $(n-1)$-spanner~\cite{Eppstein00} and therefore satisfies the dilation bound stated in the lemma.

For $k\geq 2$, consider two points $p,q \in P$.

\textbf{Case 1.} $p$ and $q$ are inside the same subtree $T$ computed in line 5 and remain in the same subtree in line 7. Let $\gamma$ be the path connecting $p$ and $q$ in the \MST. Notice that $\gamma$ is also a path in $G$. The subtree of $p$ and $q$ contains $\bigO(n/m)$ edges, which means that $\gamma$ consists of at most $\bigO(n/m)$ edges. Since $\gamma$ contains only edges of the \MST, we know that every edge of $\gamma$ has length at most $\vert pq\vert$. We can conclude that the length of $\gamma$ is at most $\bigO(n/m)\cdot \vert pq\vert$ and thereby that the dilation of $\{p,q\}$ is $\bigO(n/m)=\bigO(n/k^{d/(d-1)})$.

\textbf{Case 2.} $p$ and $q$ are inside the same subtree $T$ computed in line 5, but in different subtrees $T'$ and $T''$ after line 7, respectively. Both of these trees contain unique representatives $r', r''\in P$. Let $\gamma$ be the path connecting $p$ and $q$ in the \MST. At least one edge of $\gamma$ was removed in line 7. Let $\{p',q'\}$ be the first such edge as seen from $p$. As in case 1, the length $\vert p'q'\vert$ of $\{p',q'\}$ is at most $\vert pq\vert$. The points $p, p'$ are in the same tree $T'$ with representative $r'$. Now let $\gamma'$ be the path from $p'$ to $r'$. All edges on $\gamma'$ have length at most $\vert p'q'\vert$, since otherwise they would have been removed before $\{p',q'\}$. The path $\gamma_{p r'}$ from $p$ to $r'$ only uses edges of $\gamma$ and $\gamma'$, thus at most $\bigO(n/m)$ edges of length at most $\vert pq\vert$. The same holds for the path $\gamma_{q r''}$ from $q$ to $r''$. From this and the triangle inequality, it also follows that $\vert r'r''\vert \leq d(r', p) +\bigO(n/m)\cdot \vert pq\vert + d(q,r'') = \bigO(n/m)\cdot \vert pq\vert.$ 
The greedy 3/2-spanner on the representatives therefore contains a path of length $\bigO(n/m)\cdot \vert pq\vert$ between $r'$ and $r''$. Combining this path with $\gamma_{p r'}$ and $\gamma_{q r''}$ we obtain the necessary dilation.

\textbf{Case 3.} $p$ and $q$ are inside different subtrees computed in line 5. Let $T_p$ be the subtree of $p$ and $T_q$ be subtree of $q$ (of line 5, before line 7). Once again let $\gamma$ be the path connecting $p$ and $q$ in the \MST{} of $P$. Notice that the path $\gamma$ must contain at least two representatives, one representative of $T_p$ and one of $T_q$, since at some point $\gamma$  leaves the subtree $T_p$ through an edge removed in the construction of the subtrees and at some point enters the subtree $T_q$ through a removed edge. The vertices of the removed edges contained in $T_p$ and $T_q$ therefore must, by definition, be representatives. Let $r_p$ be the last representative of $T_p$ on the path $\gamma$ and $r_q$ be the first representative of $\gamma$ in $T_q$. Since $r_p$ and $r_q$ are both representatives, there must be a path connecting $r_p$ and $r_q$ inside the greedy $3/2$-spanner constructed on the set of representatives. The length of the shortest path between $r_p$ and $r_q$ inside the greedy $3/2$-spanner is of length at most $3/2\cdot \vert r_p r_q\vert$. We can bound the distance between $r_p$ and $r_q$ by the distance between $p$ and $q$ and the length of the paths in the \MST{} between $p$ and $r_p$ and $q$ and $r_q$ respectively. Since the subtrees $T_p$ and $T_q$ have size $\bigO(n/m)$, we know that the length of the paths between $p$ and $r_p$ and $q$ and $r_q$ in the \MST{} are at most $\bigO(n/m)\cdot \vert pq\vert$. The distance between $r_p$ and $r_q$ is therefore bounded by
\[d(r_p,p) + \vert pq \vert + d(r_q,q) = (2\cdot \bigO(n/m)+1)\cdot\vert pq\vert.\] If $p$ is in the same subtree as $r_p$ after line 7 and $q$ is in the same subtree as $r_q$, the bound on the dilation now follows. If not, we can bound the additional detour as in Case 2, which completes the case distinction and proof.
\end{proof}

From the previous lemmas we now directly obtain the main result of this section.
\upperBoundHigherDim*

\section{Lower bound} \label{sec:lower_bound}
In the previous section it was shown that \cref{alg:higherdim} gives an upper bound for the dilation of tree-width bounded spanners. In this section, we investigate corresponding lower bounds. We show that the bound given in \cref{thm:upper bound higher dim} for the dilation of a tree-width bounded spanner on Euclidean point sets is asymptotically tight: 

\lowerHigherDim*

To obtain this result, we construct a set of points resembling a grid.
While it is commonly known that a two-dimensional grid has high tree-width, the $\left(k^{1/(d-1)}\right)^d$-grid does not necessarily have tree-width $k$. We thus construct a set of points resembling the $(h+1)^d$-grid, where $h = \left\lceil (9d/2\cdot (k+2))^{1/(d-1)} -1\right\rceil$. 
To lower bound the tree-width of this grid, we generalise a basic idea for the $3$-dimensional grid  given by Korhonen in an online forum\footnote{\href{https://cstheory.stackexchange.com/questions/53029/what-is-the-treewidth-of-the-3d-grid-mesh-or-lattice-with-sidelength-n}{https://cstheory.stackexchange.com/questions/53029/what-is-the-treewidth-of-the-3d-grid-mesh-or-lattice-with-sidelength-n}}:

\begin{lemma}
\label{lem:twgrid}
The $n^d$-grid has tree-width $\geq \frac{2}{9d} \cdot n^{d-1}-1$.
\end{lemma}
\begin{proof}
To show the lemma, we  first need the definition of an \emph{all-pairs concurrent unit flow}.
Given a graph $G=(V,E)$, an \textit{all-pairs concurrent unit flow} $\Phi$ on $G$ is defined as a set of paths in $G$, so that for every pair of vertices $u,v\in V$ there is a path from $u$ to $v$ in $\Phi$. The \emph{flow at a vertex} $v$ is defined as the number of paths in $\Phi$ that pass through vertex $v$.

We now claim that for a graph $G$ of tree-width $k$ with $n$ vertices and an all-pairs concurrent unit flow on $G$, there is a vertex at which there is a flow of at least
${2n^2}/({9\cdot (k+1))}$.
Since $G$ is of tree-width $k$, there must by Lemma 2.5 in~\cite{RS86} be a separator $S$ of size $k+1$ for $G$. Let $A$ and $B$ be the induced subsets of $S$. Both of the subsets have size at least ${n}/{3}$. This means that in $\Phi$ there are atleast $2\cdot(n/3)^2$ paths between vertices of $A$ and $B$. Since there is no edge between $A$ and $B$, all these paths have to pass through the $k+1$ vertices of $S$. Therefore there must be at least one vertex $v\in S$ at which there is a flow of ${2n^2}/({9\cdot (k+1))}$.

To show that the tree-width of the $n^d$-grid is at least ${2}/(9d) \cdot n^{d-1}-1$, we construct an all-pairs concurrent unit flow $\Phi$ for which at every vertex we have a flow of at most $d\cdot n^{d+1}$. Since by our claim for any all-pairs concurrent unit flow on the $n^d$-grid there must be a vertex at which there is a flow of at least ${2n^{2d}}/({9\cdot (k+1))}$, we can conclude that the tree-width $k$ must be at least $2/(9d) \cdot n^{d-1}-1$. The construction of the all-pairs concurrent unit flow $\Phi$ is as follows: 

Given two vertices $x=(x_1,...,x_d)$ and $y=(y_1,...,y_d)$, the path from $x$ to $y$ will consist of $d$ subpaths, each parallel to one the $d$ axes. The path starts from $x$ and follows the edges parallel to the axis of the first dimension, until the point $(y_1,x_2,...,x_d)$ is reached and ends with the subpath from $(y_1,...,y_{d-1},x_d)$ to $y$. Formally we define this path as the path from $x$ to $y$ containing the edges
\begin{align*}
&\{\{(x_1+\hat{v}_1\cdot(j-1),x_2,\dots,x_d),(x_1+\hat{v}_1\cdot j,x_2,\dots,x_d)\} \mid j\in \{1, \dots,  \vert y_1-x_1\vert \}\}\\
\cup\ {} &\{\{(y_1,x_2+\hat{v}_2 \cdot (j-1),\dots ,x_d),(y_1,x_2+\hat{v}_2 \cdot j,\dots ,x_d)\} \mid j\in \{1, \dots,  \vert y_2-x_2 \vert\}\}\\
\cup\ {} &\dots\\
\cup\ {} &\{\{(y_1,\dots ,y_{d-1},x_d+\hat{v}_d \cdot (j-1)),(y_1,\dots ,y_{d-1},x_d+\hat{v}_d\cdot j)\} \mid j\in \{1, \dots,  \vert y_d-x_d\vert \}\},
\end{align*}
where
\[\hat{v}_i = \begin{cases}
1 & x_i \leq y_i\\
-1 & x_i > y_i\\
\end{cases}\]
is the direction of the $i$-th subpath.

We now argue that the flow at any vertex can be bounded by $d\cdot n^{d+1}$. Let $z$ be some vertex. For a path $\gamma$ from some $x$ to $y$ to pass through $z$ we need $z$ to be on one of the $d$ subpaths of $\gamma$ and therefore we need at least $d-1$ coordinates of $z$ to be equal to the corresponding coordinates of either $x$ or $y$. This leaves us with $d+1$ degrees of freedom for choosing $x$ and $y$ for a fixed $z$ and fixed subpath of $\gamma$. Hence the number of paths of $\Phi$ passing through $z$ is bounded by $d\cdot n^{d+1}$.
\end{proof}

Thus, the tree-width of our grid is at least 
\begin{align*}
\frac{2}{9d} \cdot (h+1)^{d-1}-1  &=  \frac{2}{9d} \cdot \left( \left\lceil (9d/2\cdot (k+2))^{1/(d-1)} -1\right\rceil+1  \right)^{d-1}-1\\
&\geq  \frac{2}{9d} \cdot \left(  (9d/2\cdot (k+2))^{1/(d-1)} \right)^{d-1}-1\\
&= k+1.
\end{align*}

Using an inductive construction for the $(h+1)^d$-grid, by taking $d$ copies of the $(h+1)^{d-1}$-grid and connecting them using $(d-1)\cdot (h+1)^{d-1}$ edges, we can prove that the $(h+1)^d$-grid has a total of $d\cdot(h+1)^d+d\cdot(h+1)^{d-1}$ edges.

The constructing of the grid-like set $P_{d,n,k}$ is as follows:
Let $m=\lfloor n/ (d\cdot (h+1)^d+ d\cdot (h+1)^{d-1})\rfloor$, be the number of points representing an edge of the grid in $P_{d,n,k}$. For every dimension $i\in \{1, \dots, d\}$, we define the set
\[P_i=\bigcup^{h}_{j_1=0}\dots \bigcup^{h}_{j_d=0}\{ (j_1\cdot m,\dots,j_{i-1}\cdot m,x_i,j_{i+1}\cdot m,\dots,j_d\cdot m) \mid x_i\in \{0, \dots, h\cdot m\}\},\]
which consists of $(h+1)^{d-1}$ sets of collinear points representing the grid edges parallel to the axis of the $i$-th dimension. Each of these collinear sets of points consists of $h \cdot m+1$ points and represents $h$ edges of the grid. Hence $\vert P_i \vert = (h \cdot m+1)\cdot (h+1)^{d-1}$ for every $i$. We can now define the set $P_{d,n,k}=\bigcup_{i=1}^d P_i$ which represents the whole $(h+1)^d$-grid. (In \cref{fig:grid} an example is shown for a $2$-dimensional point set.) Notice that the $(h+1)^d$ points $P_{\boxplus}=\{(j_1\cdot m, \dots , j_d\cdot m)\mid j_1,\dots, j_d \in \{0, \dots, h\}\}$ are contained in every $P_i$. We call these points the \textit{grid points} of the set $P_{d,n,k}$ and two grid points are called \textit{neighbouring} if the Euclidean distance between them is $m$, which corresponds to them being connected through an edge in the $(h+1)^d$-grid.

\begin{figure}[htbp]
  \centering
  \includegraphics[page=2]{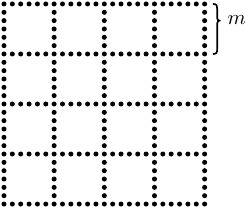}
  \caption{The set of points resembling the $(4+1)^2$-grid where $n=240$ and $m=6$.}
  \label{fig:grid}
\end{figure}

\begin{lemma}
\label{lem:UB_gridsize}
Let $d\geq 2$ be a fixed integer. The grid-like set $P_{d,n,k}$ is well-defined for every integer $n\geq d\cdot 64^{d \log d}$ and positive integer $k \leq n^{(d-1)/d}\cdot (5d)^{1/d-2}$ and is of size at most $n$.
\end{lemma}
\begin{proof}
The size of $P_{d,n,k}$ is bounded by $d\cdot (h\cdot m+1)\cdot (h+1)^{d-1} - (d-1)\cdot (h+1)^d$. This bound is obtained by adding up the number of points for every $P_i$ and subtracting the number of grid points $(d-1)$ times. Thus for the total number of points in $ P_{d,n,k}$ we obtain
\begin{align*}
\vert P_{d,n,k}\vert &= d\cdot (h\cdot m+1)\cdot (h+1)^{d-1} - (d-1)\cdot (h+1)^d\\
&= d\cdot h\cdot (h+1)^{d-1} \cdot m + d\cdot (h+1)^{d-1} - (d-1)\cdot (h+1)^d\\
&\leq d\cdot h\cdot (h+1)^{d-1} \cdot \left\lfloor \frac{n}{ d\cdot (h+1)^d+ d\cdot (h+1)^{d-1}}\right\rfloor\\
&\leq d\cdot (h+1)^{d} \cdot \frac{n}{ d\cdot (h+1)^d+ d\cdot (h+1)^{d-1}} \leq n.
\end{align*}
%
But in order for $P_{d,n,k}$ to be well defined, i.e.~resemble the $(h+1)^d$-grid, we need $m$ to be $\geq 1$. This induces, given some $d$ and a lower bound of one for $k$, a lower bound on $n$ and an upper bound on $k$. To compute appropriately small or respectively large bounds, we first bound the number of edges of the $(h+1)^d$-grid, being $d\cdot (h+1)^d+d \cdot(h+1)^{d-1}$ by
\[d\cdot (h+1)^d+d\cdot (h+1)^{d-1}\leq 2d\cdot (h+1)^d \leq 2d\left(\left(\frac{9d}{2}\cdot (k+2)\right)^{1/(d-1)}+1\right)^d.\]
In case that $k=1$, we need $n$ to be large enough so that $m\geq 1$. A lower bound for $n$ can therefore be obtained by setting $k=1$ in $2d\cdot (h+1)^d$ and then simplifying $2d\cdot (h+1)^d$ as follows:
\begin{align*}
& 2d \left(\left(\frac{9d}{2}\left(1+2\right)\right)^{1/(d-1)}+1\right)^d = 2d \left(\left(\frac{27d}{2}\right)^{1/(d-1)}+1\right)^d\\
\leq{} & 2d \left(\frac{27}{2} \cdot d +1 \right)^d \leq 2d\cdot 14^d \cdot d^d \leq d\cdot 64^{d\log d}
\end{align*}
The obtained lower bound for $n$ hence is $d\cdot 64^{d\log d}$. Given some $n\geq d\cdot 64^{d\log d}$ and some $d$ we now need an upper bound on $k$, so that $P_{d,n,k}$ is well defined, i.e.~$m\geq 1$. This bound can be obtained by solving $n\geq 2d\cdot (h+1)^d$ for $k$ which gives us

\[ n \geq 2d \left(\left(\frac{9d}{2}\left(k+2\right)\right)^{1/(d-1)}+1\right)^d \Longleftrightarrow \frac{2}{9d}\left( \left( \frac{n}{2d}\right)^{1/d}-1\right)^{d-1}-2\geq k\]

which using the lower bound on $n$ can be simplified as follows:

\begin{align*}
     \frac{2}{9d}\left( \left( \frac{n}{2d}\right)^{1/d}-1\right)^{d-1}-2 \geq \frac{2}{9d} \left( \frac{n}{3d} \right)^{\frac{d-1}{d}}-2 \geq \frac{1}{5d} \left( \frac{n}{3d} \right)^{\frac{d-1}{d}}\geq  n^{\frac{d-1}{d}} \cdot (5d)^{\frac{1}{d}-2}.
\end{align*}

\end{proof}

\begin{lemma}
If $G$ is a geometric $o\left(n/k^{d/(d-1)}\right)$-spanner on $P_{d,n,k}$, then it must be of tree-width $k+1$.
\end{lemma}

\begin{proof}
Let $G$ be a $o(n/k^{d/(d-1)})$-spanner on $P_{d,n,k}$. Let $\mathcal{R}_{x} := x+[-m/4, m/4]^d$ denote the hypercube of side-length $m/2$ centered at $x$. For any pair $p,q\in P_{\boxplus}$ of neighbouring grid points, we consider the three hypercubes $\mathcal{R}_{p}, \mathcal{R}_{q}$, and $\mathcal{R}_{s} := s+[-m/4,m/4]^d$, where $s= \frac{p+q}{2}$ is the midpoint between $p$ and $q$. We argue that there is a path in $G$ from $p$ to $q$ such that: \emph{(i)} the path only visits vertices in  $\mathcal{R}_{p, q}:= \mathcal{R}_{p} \cup \mathcal{R}_{q} \cup \mathcal{R}_{s}$ and \emph{(ii)} for any $p' \in \mathcal{R}_{p}$ and $q' \in \mathcal{R}_{q}$, if the path visits $p'$ and $q'$ then it visits $p'$ before it visits $q'$. 

Let $p_0 = p, p_1, p_2, \ldots p_m=q$ the ordered sequence of points representing the edge of the grid between $p$ and $q$, and further let $s_i := \frac{p_i + p_{i+1}}{2}$.  Since $G$ is a $o(n/k^{d/(d-1)})$-spanner, we can assume that there is for every pair $p_i$, $p_{i+1}$ a shortest path $\gamma_i$ in $G$ of length at most $m/4 = \Theta(n/k^{d/(d-1)})$. Such a path necessarily will remain with in $\mathcal{R}_{p_i}\cap \mathcal{R}_{p_{i+1}} \subset \mathcal{R}_{s_i}$.
For $0 \leq i < m/2$ we have $\mathcal{R}_{s_i} \subset  \mathcal{R}_{p} \cup \mathcal{R}_{s}$, and for $m/2 \leq i <m$ we have $\mathcal{R}_{s_i} \subset  \mathcal{R}_{q} \cup \mathcal{R}_{s}$.

Consider the walk from $p$ to $q$ obtained by concatenating the paths $\gamma_0, \gamma_1, \ldots \gamma_{m-1}$. It first only visits vertices in $\mathcal{R}_{p} \cup \mathcal{R}_{s}$ and then only vertices in $\mathcal{R}_{q} \cup \mathcal{R}_{s}$. By removing cycles we obtain a path $\gamma_{p,q}$ from $p$ to $q$ with this property. Thus, we have a path with the properties \emph{(i)}$+$\emph{(ii)} stated above.

We have such a path $\gamma_{p,q}$ for any pair $p,q\in P_{\boxplus}$.
We claim that this implies that $G$ must have a minor isomorphic to the $(h+1)^d$-grid and is therefore of tree-width $>k$. To show this, we perform a series of path contractions on $G$, obtaining a graph $G'$. More specifically, we first contract the edges of $\gamma_{p,q}$ starting from $p$, until we encounter as endpoint the last vertex in $\mathcal{R}_{p}$, identifying the merged vertices with $p$. After skipping an edge, we contract the remaining edges of $\gamma_{p,q}$, identifying the merged vertices with $q$. This is illustrated in \cref{fig:newgrid}. After these contractions, the remaining edge of $\gamma_{p,q}$ is an edge from $p$ to $q$.

Further we observe that \emph{(i)} any vertex $p'$ of $\gamma_{p,q}$ in $\mathcal{R}_{p}$ has been identified with $p$, \emph{(ii)} any vertex $q'$ of $\gamma_{p,q}$ in $\mathcal{R}_{q}$ has been identified with $q$, and \emph{(iii)} any other vertex $p''$ of $\gamma_{p,q}$ does not also lie on another path $\gamma_{r,s}$ with $\{p, q\} \neq \{r, s\}$. The latter holds, since $p''$ cannot lie in $\mathcal{R}_{r, s}
$. This means, that the paths $\gamma_{p,q}$ and $\gamma_{r,s}$ are contracted consistently.
Thus, regardless of the order in which we contract the paths, we obtain a graph $G'$ which has a minor isomorphic to the $(h+1)^d$-grid. In $G'$, there is an edge between every neighbouring pair of grid points. Therefore the subgraph $G'[P_{\boxplus}]$ containing only the grid points is, if we remove edges between non-neighbouring grid points, isomorphic to the $(h+1)^d$-grid. 
We conclude that the spanner $G$ has a minor isomorphic to the $(h+1)^d$-grid and by Lemma~\ref{lem:twgrid} therefore must be of tree-width at least $\frac{2}{9d}(h+1)^{d-1}-1\geq k+1$. 
\end{proof}

\begin{figure}[htbp]
  \centering
  \includegraphics[page=4]{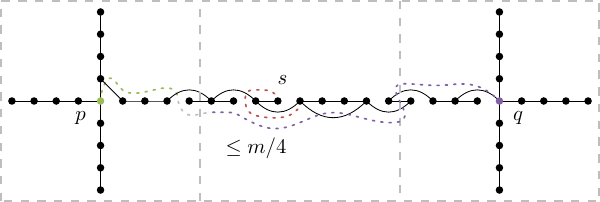}
  \caption{Grid points $p$ and $q$ as well as the three hypercubes centered at $p$, $q$ and $s$. The subpath of $\gamma_{p,q}$ that is contracted to $p$ is shown in green and the subpath that is contracted to $q$ in purple. The red path is a detour of length $\leq m/4$ for points $p_i$ and $p_{i+1}$.}
  \label{fig:newgrid}
\end{figure}

\section{Minor-3-cores and the tree-width of planar spanners}\label{sec:planar}
We now have given an asymptotically tight algorithm to obtain small bounded tree-width, with a dilation dependent on the chosen tree-width. 
We now focus on bounded tree-width spanners with the additional property of being plane. 

In general, there are graphs with $m$ edges and tree-width at least $m \cdot \epsilon$ for a constant $\epsilon$ \cite{Grohe2009}. We show that this is not true for conneced planar graphs, yielding a tight dependency between tree-width and the number of edges:

\twtreeaddedges*

\subsection{Minor-3-cores} \label{subsec:minor3cores}

To obtain this Theorem, we first have to introduce the concept of \textit{minor-3-cores}, in reference to the $k$-core which was established by Seidman~\cite{Seidman1983} for the study of social networks. We then use this concept to obtain the tree-width of planar graphs with a certain fixed number of edges. 

\begin{definition}[Minor-3-Core]
Let $G$ be a graph and $H$ a minor of $G$ with minimum degree $3$. If there is no minor of minimum degree $3$ that has more edges than $H$ i.e., $H$ is an edge maximal minor with this property, we call it a \textit{minor-3-core} of $G$.
\end{definition}

This definition gives a very clear idea of what a minor-3-core should be, but unfortunately it does not lead to any structural insight concerning the concept. 
It especially leaves open whether minor-3-cores are unique or if there are multiple non-isomorphic minor-3-cores. 
We therefore need to further characterise minor-3-cores with regard to algorithmic properties and uniqueness. In order to answer these questions, we first define the disjoint-paths property. 

\begin{definition}[Disjoint-Paths Property]
Let $G$ be a graph and let $C$ be a minor-3-core of $G$ as well as $v$ be a vertex of the minor-3-core (and therefore a vertex of $G$ that has not been contracted). We say that $v$ fulfils the \textit{disjoint-paths property} for $C$, if for every neighbour $u_1,...,u_{\Delta_C (v)}$ of $v$ in $C$, there is a path $\gamma_i$ in $G$ which leads from $v$ to $u_i$ and $\left(\bigcup_{u \text{ in } \gamma_i} u \right) \cap \left( \bigcup_{w \text{ in } \gamma_j} w \right) = \{v\}$ for every path $\gamma_j$ for $j\neq i$.
\end{definition}

We now show that we can assume having the disjoint paths property in a minor-3-core:

\begin{lemma}
\label{lem:m3cdpp}
For every graph $G$ and minor-3-core $C$ of $G$, there is a minor-3-core isomorphic to $C$ in which every vertex has the disjoint paths property.
\end{lemma}

\begin{proof}
Let $C$ be a minor-3-core of a graph $G$. As $C$ is a minor of $G$, there is a series of contractions on a subgraph $G'$ of $G$ that results in $C$. 
The contractions on $G'$ induce a partition of $G'$, in which every set can be associated with a vertex of $C$. Two vertices $u,v\in V(G')$ are said to be in the same set, if in the construction of $C$ they are (transitively) contracted to the same vertex. More formally: Let $G'_\ell$ be the graph $G'$, where we label every vertex of $G'$ by a singleton set containing a vertex of $G'$, i.e. $G'_\ell=(\{\{v\}\mid v\in V(G')\}, \{\{\{u\},\{v\}\}\mid \{u,v\}\in E(G')\})$. If we perform a contraction of an edge $\{U,W\}$ on this graph, we label the resulting vertex $U\cup W$. Now, if we perform the contractions of $G'$ that result in $C$ on $G'_\ell$, we end up with an graph $C_\ell$ isomorphic to $C$, where every vertex $U$ is the set of vertices of $G'$ that were (transitively) contracted to form the vertex $U$, i.e. $U$ is a set of the partition.

We now want to show that there is exactly one \emph{canonical vertex} in every $U \in V(C_\ell)$, denoted by $u_c$, such that $u_c \in U$, $\Delta_{G'}(u_c) \geq 3$ (and thus $\Delta_{G}(u_c) \geq 3$) and that fulfils the disjoint path property for $C_\ell'$, where by $C_\ell'$ we denote a graph isomorphic to $C_\ell'$ where every $U \in V(C_\ell)$ is replaced by its canonical vertex. 
Let $W_1,...,W_{\Delta (U)}$ be the neighbours of $U$. For every $W_i$ there is a path in $G$ starting at $u_c$ that leads to a vertex of $W_i$ and that is disjoint (except for the vertex $u_c$ itself) to every path starting at $u_c$ and leading to a vertex of $W_j$ for $j\neq i$. We argue that there is exactly one vertex for every $U\in V(C_\ell)$ with this property. 

Assume that there is no single vertex with this property. Then there must be a minimal set of vertices in $U$, so that there are disjoint (except for the starting vertices) paths in $G$ to vertices of the sets that are neighbours of $U$. We can assume that the vertices of this minimal set are at least of degree 3 and that there are at least two (almost) disjoint paths starting at each of these vertices. Let $u_1$ and $u_2$ be a pair of these vertices, such that there is a path inside $G[U]$ that connects $u_1$ and $u_2$ and that is disjoint to the paths to the vertices of the sets that are neighbours of $U$. Notice that the vertices on this path were contracted in the construction of $C$ and that contracting the same vertices except for $u_1$ and $u_2$, would increase the number of edges of $C$ by one, while every vertex is still of degree at least 3. This is clearly a contradiction to $C$ being a minor-3-core and we can conclude that there must be exactly one vertex in $U$ that has the disjoint-paths property.
\end{proof}

By \cref{lem:m3cdpp} we now know that there is a minor-3-core, in which every vertex has the disjoint-paths property. We call a minor with this property a \emph{canonical minor-3-core}:

\begin{definition}[Canonical Minor-3-Core]
Let $C$ be a minor-3-core of a graph $G$. We call $C$ the canonical minor-3-core, if every vertex of $C$ has the disjoint-paths property. The vertices of $C$ are called the canonical vertices.
\end{definition}

We now show that the canonical minor-3-core is, as the name implies, unique: 

\begin{lemma}\label{lem:can-m3c-unique}
Let $G$ be a graph, the canonical minor-3-core of $G$ is unique.
\end{lemma}
\begin{proof}
W.l.o.g. we assume $G$ to be connected. (Otherwise we would simply look at every connected component individually.) 
Suppose that there are two distinct canonical minor-3-cores $C_1$ and $C_2$. We claim that there must be a minor of $G$ that is of minimum degree $\geq 3$ and has more edges than $C_1$ and $C_2$. This minor $H$ is obtained by recursively contracting every vertex of $G$ to one of its neighbours except for the vertices $V(C_1)\cup V(C_2)$. We first prove that every vertex in $H$ must be of degree $\geq 3$. Let $v$ be a vertex of $C_1$. For every neighbour of $v$ in $C_1$, it holds that there either is an edge to the neighbour itself or on the path connecting $v$ and its neighbour in $G$ there is a vertex of $C_2$. If $u$ is the first vertex of $C_2$ on this path, then instead of having an edge to its neighbour in $C_1$, $v$ has an edge to $u$. Analogously, this holds for vertices in $C_2$. Hence, every vertex in $H$ has degree $\geq 3$. We now argue that $H$ has more edges than $C_1$ and $C_2$. Notice that $C_1$ must be a minor of $H$, since $C_1$ can be obtained by contracting every vertex of $H$ to one of its neighbours except for the vertices of $C_1$. The same holds for $C_2$ as well. Since we can assume that $C_1$ and $C_2$ are edge maximal, they have to differ by at least one vertex. Therefore, $H$ has at least one vertex more than $C_1$ and $C_2$ and since every vertex is of degree $\geq 3$, $H$ must have more edges than $C_1$ and $C_2$. This is a contradiction to $C_1$ and $C_2$ being minor-3-cores, which proves the statement of the lemma.
\end{proof}

Combining \cref{lem:m3cdpp} and \cref{lem:can-m3c-unique}, we can state that every minor-3-core $C$ of $G$ is isomorphic to the canonical minor-3-core of $G$. Thus we obtain uniqueness and can refer to \textit{the minor-3-core}, instead of a minor-3-core, meaning the canonical minor-3-core.

We will now give an algorithmic characterisation for the minor-3-core. Similarly to $k$-cores~\cite{Seidman1983} we can compute the minor-3-core by an iterative pruning process:

\begin{theorem}
Let $G$ be a graph. The minor-3-core of $G$ can be obtained by performing the following operations to exhaustion: 
\begin{itemize}
\item Let $e$ be an edge incident to a vertex $v$ with degree $\leq 2$. Contract $v$ along $e$.
\item If $v$ is an isolated vertex, delete $v$.
\end{itemize}
\end{theorem}
\begin{proof}
The graph $H$ resulting from the iterative pruning process described above is clearly a minor of $G$, in which every vertex has degree at least 3. Let $C$ be the canonical minor-3-core of $G$. Suppose that there is some edge in $C$ that is not in $H$. For an edge to not be an edge of $H$ one of the vertices of this edge needs to be contracted to one of its neighbours in the construction of $H$. Let $u$ be the first vertex of $C$ that is contracted in the construction of $H$. By definition of the pruning process $u$ needs to be of degree $\leq 2$ to be contracted. Since $u$ by the definition of the disjoint-paths property starts off as a vertex of degree $\geq 3$, there needs to be a contraction that reduces the degree of $u$ to $\leq 2$. For this to be possible, there needs to be a contraction of two neighbouring vertices of $u$, where both of these vertices are canonical vertices. This must be the case, since in $G$ there are at least three disjoint paths starting in $u$ that lead to at least three different canonical vertices and there cannot be an edge between the vertices of the disjoint paths, since otherwise these vertices would be canonical vertices themselves. The vertices of this edge would need to be canonical as they are both connected to $u$, to one other canonical vertex each, as well as to each other by disjoint paths. But since we assumed $u$ to be the first canonical vertex that gets contracted in the construction of $H$, we reach a contradiction. Hence every edge of $C$ is also an edge of $H$ and since $C$ is edge maximal, $H$ must be equal to the canonical minor-3-core $C$.
\end{proof}

In the following we use the minor-3-core in order to bound the tree-width of planar spanners. To begin the exploration of the relationship of minor-3-cores and tree-width we first make a couple of observations.

\begin{enumerate}
\item The minor-3-core of a tree is the empty graph.
\item The minor-3-core of a graph with at most one circle is the empty graph.
\item The $K_4$ is its own minor-3-core. 
\end{enumerate}

Even more:

\begin{remark}
\label{remark1}
Every graph without a $K_4$ minor has an empty minor-3-core.
\end{remark}

\begin{proof}
A graph has no $K_4$ minor if and only if every biconnected component is a series parallel graph. The graph can therefore be decomposed into a tree of series parallel graphs. Starting at the leaves we can recursivly contract the biconnected components to a single vertex of the parent biconnected component, by contracting the vertices of degree $\leq 2$ in the series parallel biconnected component. We end up with the root biconnected component, which of course is series parallel and therefore has an empty minor-3-core.
\end{proof}

Considering a graph and then connecting a tree to some vertex of our graph, like a branch growing out of our graph, does not increase the tree-width of the graph. The same holds for edges that we replace with paths. So we can see that it is possible to increase the size of our graph without affecting the tree-width of the graph. This is a problem if we want to bound the tree-width of our graph through common separator arguments, for example using~\cite{Robertson1994}, since the bound on the tree-width there depends on the size of the graph. Hence it would be nice to have a tool that prunes our graph by removing “tree-like” structures from our graph. For that, we  use the definition of minor-3-cores: 

\begin{theorem}
\label{lem:3coretw}
Given a graph of tree-width $k\geq 3$, the tree-width of the minor-3-core of $G$ is $k$ as well.
\end{theorem}
\begin{proof}
Let $G$ be a graph of tree-width $k\geq 3$ and $H$ its minor-3-core. Since $H$ is a minor of $G$, the tree-width of $H$ can be at most $k$. To prove that $H$ must be of tree-width at least $k$, we prove that contracting a vertex of degree at most $2$ does not decrease the tree-width. So let $G$ be a graph of tree-width $k\geq 3$. Suppose that there is a vertex $v$ of degree $2$ and let $u_1$ and $u_2$ be the two vertices adjacent to $v$. Let $G'$ be the graph resulting from the contraction of $v$ to $u_1$. We now argue that $G'$ must have tree-width $k$ as well. Assume that $G'$ has tree-width $k'<k$ and let $(T',X')$ be a tree-decomposition of width $k'$ for $G'$. We now construct a tree-decomposition $(T,X)$ for $G$ of width $k'$ from $(T',X')$. In $(T',X')$ there must be some bag $X(i)$ which contains both $u_1$ and $u_2$, since they are connected by an edge. We now create a new node, whose bag contains $v$, $u_1$ and $u_2$ and connect it to $i$ by an edge. The tree-decomposition $(T,X)$ is now a valid tree-decomposition for 
$G$ 
with tree-width $k'<k$, which is a contradiction to $G$ being of tree-width $k$. If $v$ is of degree $1$, the same argument holds, but we only consider $u_1$, which has to appear in at least one bag of $(T',X')$.
\end{proof}

As mentioned before we want to bound the tree-width of graphs (especially spanners) through the size of their minor-3-core. Therefore it would be useful to have a bound on the size of the minor-3-core of a graph. If a graph is very “tree-like”, i.e. the graph can be decomposed into a set of trees by removing only a few edges, we expect it to have a small minor-3-core. We formalise this statement as follows:
\begin{lemma}
\label{lem:3coresize}
Let $G$ be a connected graph with $n$ vertices and $n-1+m$ edges, where $m\geq 1$ is some integer. The minor-3-core of $G$ has at most $2\cdot (m-1)$ vertices.
\end{lemma}
\begin{proof}

We prove the lemma by induction on $m$. If $m=1$ then $G$ can contain at most one circle. The minor-3-core of $G$ therefore must be the empty graph. Now let $G$ be a graph with $n-1+m-1$ edges, where $m> 1$. By the induction hypothesis the minor-3-core of $G$ has at most $2\cdot (m-2)$ vertices. If we insert an edge $e$ into $G$, only the degree of the two incident vertices of $e$ will increase. We claim that in the pruning process constructing the minor-3-core at any point in time there are only two vertices, whose degree increased by the insertion of $e$. This holds true, since by contracting a vertex $u$ along an edge $\{u,v\}$ only the degree of $v$ may increase and the vertex $u$ is deleted. Therefore the number of vertices that were spared from being contracted in the pruning process constructing the minor-3-core after the insertion of $e$, is at most two. Hence the minor-3-core of $G$ after the insertion of $e$ has at most $2\cdot (m-2)+2=2\cdot (m-1)$ vertices.
\end{proof}

This lemma is the last part we need to prove an upper bound for the tree-width of connected planar graphs. Intuitively the statement of the lemma could just as well be stated as a bound on the size of minor-3-cores of trees with an additional $m$ edges. If we place these additional edges in a planar fashion, we can combine our bound on the size of minor-3-cores and the bound on the tree-width of planar graphs derived from Theorem 6.2 in~\cite{Robertson1994} to obtain a bound on the tree-width of trees with additional edges that were placed without creating edge crossings.
Using this, we can finally proof our second main result: 

\begin{figure}[htbp]
  \centering
  \includegraphics[page=6]{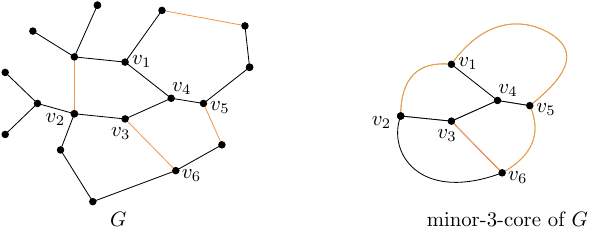}
  \caption{Minor-3-core of a graph $G$ that consists of a tree with 4 additional edges (orange).}
  \label{fig:m3c}
\end{figure}

\begin{proof}[Proof of \cref{thm:twtreeaddedges}]
Since $G$ is a connected graph with $n+\lfloor (k-1)^2/72 \rfloor$ edges, its minor-3-core must by Lemma \ref{lem:3coresize} have at most $\lfloor (k-1)^2/36 \rfloor$ vertices. Also, since the minor-3-core is a minor of $G$ and $G$ is planar, the minor-3-core must be planar as well. The tree-width of a planar graph with $n$ vertices by~\cite{Robertson1994} is at most $6\sqrt{n}+1$. The tree-width of the minor-3-core of $G$ therefore is
\[6\sqrt{\left\lfloor (k-1)^2/36 \right\rfloor}+1\leq 6\sqrt{(k-1)^2/36 }+1\leq k.\]
By Lemma \ref{lem:3coretw} we can now conclude that the tree-width of $G$ is either $\leq 2$ or exactly $k$, so in both cases at most $k$.
\end{proof}

\subsection{Tree-width of planar spanners} \label{subsec:tw-planar}

\cref{thm:twtreeaddedges} allows us a stronger result than \cref{thm:upper bound higher dim} in the Euclidean plane: Adjusting the algorithm given in \cite{Aronov2008}, we obtain a plane bounded-degree spanner with bounded tree-width:

\corPlaneTreewidthSpanner*

Our main adjustment to the algorithm is to chose a plane constant-dilation spanner $\mathcal{S}$ of $P$ with maximum degree $4$, as provided in \cite{kanj.2017}, instead of the Delaunay triangulation and to utilise a minimum spanning tree of $\mathcal{S}$ instead of the Euclidean MST.

\begin{algorithm}[H] \caption{$\texttt{PlaneBoundedTreeWidthSpanner}(P,k)$}
\label{alg:2D}
\begin{algorithmic}[1]
\Require a set of $n$ points $P\subset \mathbb{R}^2$ and a positive integer $k\leq 12\sqrt{n-3}$
\Ensure a geometric spanner $G=(P,E)$ of tree-width $k$ and dilation $\mathcal{O}(n/k^2)$
\State $\mathcal{S} \leftarrow $ Plane $20$-spanner on $P$ with max degree $4$
\State $\mathcal{MST} \leftarrow $ minimum spanning tree of $\mathcal{S}$
\If {$k=1$}
  \State { \Return $\mathcal{MST}$}
\EndIf
\State $m \leftarrow \left\lfloor (k-1)^2/144 \right\rfloor + 3$
\State Compute $m$ disjoint subtrees $\mathcal{T}$ of $\mathcal{MST}$, each containing $\mathcal{O}(n/m)$ points.
\State $E\leftarrow \bigcup_{T\in\mathcal{T}}E(T)$
\ForEach {$T, T'\in \mathcal{T}$ where $T\neq T'$}
 \If {there is some $\{p,q\}\in E(\mathcal{S})$ where $p\in V(T)$ and $q\in V(T')$}
  \State {Add the shortest edge $\{p',q'\}\in E(\mathcal{S})$ where $p'\in V(T)$ and $q'\in V(T')$ to $E$.}
\EndIf
\EndFor
\State \Return $G=(P,E)$
\end{algorithmic}
\end{algorithm}

Since the spanner $G$ constructed by \cref{alg:2D} is a subgraph of $\mathcal{S}$, it is plane and has maximum vertex degree $4$.
The total number of edges of $G$ is at most
\[n-1-(m-1)+3m-6=2m-6\leq n+\left\lfloor \frac{(k-1)^2}{72}\right\rfloor ,\] since between every pair of subtrees there is at most one edge and the graph obtained by contracting every subtree of $G$ to a single vertex is planar. The tree-width therefore follows from \cref{thm:twtreeaddedges}. 

For the dilation, we use Lemma 4 of \cite{Aronov2008}, which states that for the MST of an arbitrary $t$-spanner, every edge on the path in the MST between two points $p$ and $q$ has length $\leq t \cdot |pq|$. 
We consider $p$ and $q$ such that there is an edge between $p$ and $q$ in $\mathcal{S}$, but not in $G$. If $p$ and $q$ are in the same subtree, then there is a path of length $t \cdot |pq| \cdot \mathcal{O}(n/m)$ between $p$ and $q$. If $p$ and $q$ are in different subtrees, then we consider a path in $\mathcal{MST}$ between $p$ and $q$. Let $r$ be the last vertex in the subtree of $p$ and $r'$ the first vertex in the subtree of $q$ on this path. The two subpaths from $p$ to $r$ and from $r'$ to $q$ contain $\bigO(n/m)$ edges, which have by \cite{Aronov2008} at most length $t\cdot |pq|$. Let $\{p',q'\}$ be the shortest edge in $G$ between both subtrees. Since $|rr'| \leq |pq|$, all edges on the path from $r$ to $p'$ and all edges on the path from $q'$ to $r'$ have length $\leq t \cdot |pq|$. Thus we can construct a detour of length at most $d(p,r)+d(r,p')+|p'q'|+d(q',r')+d(r',q)=\bigO(n/m) \cdot |pq|$.
We then obtain an overall dilation of $\bigO(n/m)$.

Since $\mathcal{S}$ and its MST can be computed in $\mathcal{O}(n \log n)$ \cite{kanj.2017}, and the computation of the subtrees runs in $\bigO(n \log n)$ by \cite{Aronov2008}, \cref{alg:2D} has a total running time of $\bigO(n \log n)$.

\section{Bounded tree-width spanners on convex point sets} \label{sec:convex}

We have seen that for unrestricted sets of points in the plane the Delaunay triangulation is an $1.998$-spanner of tree-width $6\sqrt{n}+1$. To obtain spanners with smaller tree-widths, we now look at sets of points in the plane that are convex or even lie on a circle - and by that also answer an open question stated in~\cite{Aronov2008}.

\subsection{Tree spanners}
The Euclidean minimum spanning tree is an $(n-1)$-spanner for any set of $n$ points in $\mathbb{R}^d$ of tree-width $1$. Aronov et. al.~\cite{Aronov2008} proved that there are sets of points, on which the dilation bound of the minimum spanning tree is almost optimal. To be more precise they showed that there is a set of points, on which any tree-spanner has dilation at least $\frac{2n}{\pi}-1$. The set they used in their proof was a set of $n\geq 3$ points equally spaced on the unit circle. So, even for tree-spanners on convex sets and sets of points on a circle a dilation of $\Omega(n)$ may be unavoidable. The most obvious tree-spanner on a set of points equally spaced on the unit circle would probably be the spanner that consists of all the edges between points next to each other except for one edge. It is not hard to see that this spanner has dilation $n-1$. But we can show that for points equally spaced on the unit circle, there are tree-spanners with dilation almost $\frac{2n}{\pi}-1$. This answers the open question stated in~\cite{Aronov2008}.

\begin{lemma}
\label{lem:sawtoothspanner}
Let $P$ be a set of $n$ points equally spaced on a circle. There is a tree-spanner on $P$ with dilation $\leq \frac{2n}{\pi}+\frac{\pi}{2n}$.
\end{lemma}

\begin{proof}

Given some $n\geq 2$, let $p_i=\left(\sin \left(\frac{2\pi i}{n}\right), \cos\left(\frac{2\pi i}{n}\right)\right)$. We define $P=\left\{ p_i  \mid i \in \{0, \dots, n-1\}\right\}$ which is the set of $n$ points equally spaced on the unit circle, where one point is $(0,1)$. Notice that for any $i$ we have $p_i=p_{i+n}$.

We now construct a tree-spanner $G$ for $P$ with dilation $\leq \frac{2n}{\pi}+\frac{\pi}{2n}$. The spanner $G$ is a path visiting each point on the unit circle in a sawtooth pattern. We define $G$ as the path $p_0,p_{n-1},p_1,\dots , p_{\lfloor n/2 \rfloor}$ or more formally as the graph with edges
\[E=\{\{p_i,p_{n-1-i}\} \mid i\in \{0, \dots,  \lfloor n/2\rfloor-1 \}\} \cup \{\{p_{i+1},p_{n-i}\} \mid i\in \{0, \dots, \lfloor (n-1)/2\rfloor-1 \}\}.\]

\begin{figure}[htbp]
  \centering
  \includegraphics[width=250pt]{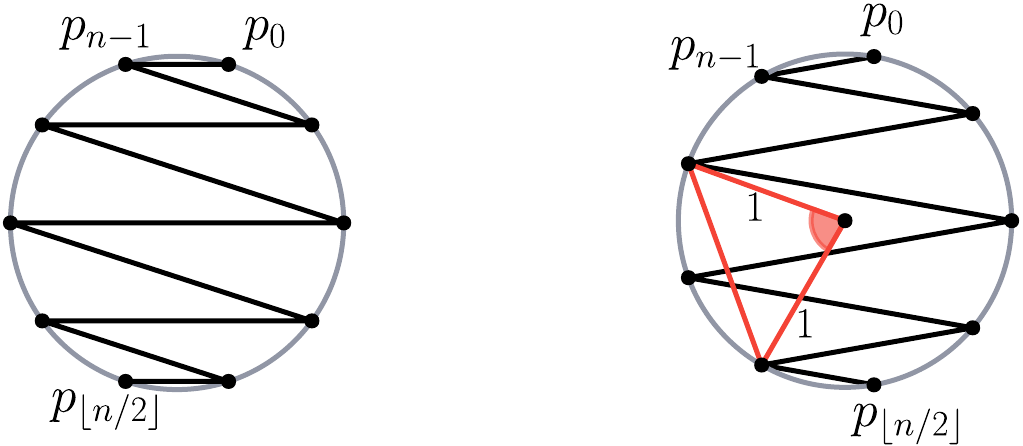}
  \caption{The spanner $G$ on $n$ points equally spaced on the unit circle for $n=10$ (left) and for $n=9$ (right).}
\end{figure}

The Euclidean distance between two points $p_i$ and $p_j$ can be derived from the law of cosines:

\[\vert p_i p_j\vert = \sqrt{2-2 \cos \left(\frac{2\pi}{n} (j-i) \right) } = \left\vert 2 \sin \left((j-i)\cdot \frac{\pi}{n}  \right) \right\vert .\]

For the distance in $G$ between two points we can add up the Euclidean distances between the points on the path connecting them. Since the spanner $G$ is a path from $p_0$ to $p_{\lfloor n/2\rfloor}$,  we now let $q_1,...,q_n$ be the vertices along the path $G$, where $q_1=p_0$ and $q_{n}=p_{\lfloor n/2 \rfloor}$. The distance in $G$ from some $q_i$ to $q_j$ for $i\leq j$ then is
\[d(q_i,q_j)=\sum_{k=i}^{j-1} 2\sin \left( \frac{\pi}{n}\cdot k\right).\]

To calculate the Euclidean distance between two points $q_i$ and $q_j$ we cannot simply use the formula given above, since it depends on the indices $i'$ and $j'$ of the corresponding points $p_{i'}=q_i$ and $p_{j'}=q_j$. To derive a formula for the Euclidean distance given some $q_i$ and $q_j$ we first make the following two observations. For $i\in \{1, \dots,  n\}$ we have that if $i$ is odd then $q_{i}=p_{\lfloor i/2 \rfloor}$ and if $i$ is even then $q_{i}=p_{n-\lfloor i / 2\rfloor}$. Notice that for $i\leq \lfloor n/2\rfloor$ and $j > \lfloor n/2\rfloor$ it holds that $\angle p_ip_j = \angle p_ip_{n-j}$. From this we can deduce that $\vert q_iq_j\vert = \vert p_{\lfloor i/2\rfloor} p_{\lfloor j/2\rfloor}\vert$. The Euclidean distance between two vertices $q_i$ and $q_j$ therefore is
\[\vert q_i q_j\vert = \left\vert 2 \sin \left( \left\lfloor \frac{j-i+1}{2}\right\rfloor \cdot \frac{\pi}{n} \right)\right\vert.\]

Finally, to prove Lemma \ref{lem:sawtoothspanner} we still need two trigonometric identities. These two identities are the Lagrange trigonometric identity~\cite{Jeffrey2008}:
\begin{equation}
\label{eq:lagrange}
\sum_{k=0}^n 2\sin (k\theta) = \frac{\cos \left(\frac{1}{2}\theta \right)-\cos \left(\left(n+\frac{1}{2}\right)\theta\right)}{\sin \left(\frac{1}{2}\theta \right)}
\end{equation}
and the product-to-sum identity~\cite{Abramowitz1988}:
\begin{equation}
\label{eq:sum-to-product}
2 \sin \theta \sin \varphi = \cos(\theta - \varphi) - \cos (\theta + \varphi).
\end{equation}

To prove the lemma we first simplify the formula for the distance between points $q_i$ and $q_j$ in $G$ and their Euclidean distance using the trigonometric identities above. We assume that $i\leq j$.

\begin{align*}
d(q_i,q_j)=\ & \sum_{k=i}^{j-1}2\sin\left(\frac{\pi}{n}\cdot k\right)=\sum_{k=0}^{j-1}2\sin\left(\frac{\pi}{n}\cdot k\right)-\sum_{k=0}^{i-1}2 \sin\left(\frac{\pi}{n}\cdot k\right)\\
\stackrel{\text{(\ref{eq:lagrange})}}{\leq} & \ \frac{ \cos\left(\frac{\pi}{2n}\right) -\cos\left( \left(j-1+\frac{1}{2}\right)\cdot\frac{\pi}{n} \right) - \cos\left(\frac{\pi}{2n}\right) +\cos\left( \left(i-1+\frac{1}{2}\right)\cdot\frac{\pi}{n} \right)  }{\sin \left(\frac{\pi}{2n}\right)}\\
\leq& \ \frac{  \cos\left( \left(i- \frac{1}{2}\right) \cdot \frac{\pi}{n} \right) -\cos\left( \left(j- \frac{1}{2}\right) \cdot \frac{\pi}{n} \right)  }{\sin \left(\frac{\pi}{2n}\right)}
\end{align*}

And for the Euclidean distance, assuming $i\leq j$, we have:

\begin{align*}
\vert q_i q_j\vert& = \ 2 \sin \left(\left\lceil \frac{j-i+1}{2}\right\rceil\frac{\pi}{n}\right) \geq 2 \sin \left( \frac{j-i}{2}\cdot \frac{\pi}{n}\right)\\
&\stackrel{\text{(\ref{eq:sum-to-product})}}{=} \ \frac{\cos \left( \left(i-\frac{1}{2}\right)\cdot \frac{\pi}{n}\right) - \cos \left( \left(j-\frac{1}{2}\right)\cdot \frac{\pi}{n}\right)}{\sin \left( (j+i-1)\cdot \frac{\pi}{2n} \right)} \\
&\geq \cos \left( \left(i-\frac{1}{2}\right)\cdot \frac{\pi}{n}\right) - \cos \left( \left(j-\frac{1}{2}\right) \cdot \frac{\pi}{n}\right).
\end{align*}
The dilation of the spanner $G$ is therefore bounded by
\[\frac{d(q_i,q_j)}{\vert q_i q_j \vert} \leq \frac{  \cos\left( \left(j- \frac{1}{2}\right) \cdot \frac{\pi}{n} \right) -\cos\left( \left(k- \frac{1}{2}\right) \cdot \frac{\pi}{n} \right)  }{  \sin \left(\frac{\pi}{2n}\right) \cdot   \left(\cos\left( \left(j- \frac{1}{2}\right) \cdot \frac{\pi}{n} \right) -\cos\left( \left(k- \frac{1}{2}\right) \cdot \frac{\pi}{n}\right) \right)    }              = \frac{1}{\sin \left( \frac{\pi}{2n}\right)} \leq \frac{2}{\pi}n+\frac{\pi}{2n} ,\]
which proves Lemma \ref{lem:sawtoothspanner}.
\end{proof}

\subsection{Spanners of tree-width 2}

While for geometric spanners of tree-width $1$, there are sets of points in convex position on which a dilation of $\Omega (n)$ is unavoidable, we can obtain some results for geometric spanners of tree-width $2$ on sets of points in convex position.

As every plane graph on a set of points in convex position is outerplanar, and outerplanar graphs have tree-width $\leq 2$~\cite{Dinneen1995}, every plane spanner on a set of points in convex position has tree-width $\leq 2$.
Let $P$ be a point set with points in convex position. Since by~\cite{Amani2016}, there is a plane $1.88$-spanner for $P$, there is a $1.88$-spanner of tree-width $2$ for $P$.  
If the points are equally spaced on a circle, by~\cite{Sattari2019} there is even a $1.4482$-spanner with tree-width $2$.

To show a lower bound, we build on the idea of~\cite{Dumitrescu2016}:

\begin{lemma}
\label{lem:spannerconvexhull}
Let $P$ be the set of $n\geq 3$ points equally spaced on the unit circle. Every spanner on $P$ with dilation $< 2$ must contain the convex hull of $P$.
\end{lemma}
\begin{proof}
Suppose $G$ is $t$-spanner of $P$ for some $t<2$ and there are points $p$ and $q$ in $P$ that share an edge in the convex hull of $P$ but not in $G$. The Euclidean distance between $p$ and $q$ is $\sin(\pi/n)$. The shortest detour from $p$ to $q$ has length $\sin(\pi/n)+\sin(2\cdot\pi/n)$. It holds that
\[\frac{\sin(\pi/n)+\sin(2\cdot \pi/n)}{\sin(\pi/n)}\geq 2.\]
This is a contradiction to our assumption $t<2$.
\end{proof}

\begin{lemma}
\label{lem:tw2spannerplane}
Let $P$ be the set of $n\geq 3$ points equally spaced on the unit circle. Every $t$-spanner for $t<2$ of tree-width $2$ on $P$ is plane.
\end{lemma}
\begin{proof} 
Let $G$ be a $t$-spanner for $t<2$ of tree-width $2$ on $P$. By \cref{lem:spannerconvexhull}, $G$ contains the convex hull of $P$. For $n=3$ the statement holds, since the $K_3$ is plane and has tree-width $2$. So let $n\geq 4$ and suppose in the straight-line drawing of $G$ there is an edge-crossing. Let these edges be $\{p,q\}$ and $\{p,q'\}$. On the convex hull of $P$ there must be a path from $p$ to $p'$, not containing $q$ and $q'$, a path from $p'$ to $q$, not containing $p$ and $q'$, a path from $q$ to $q'$ not containing $p$ and $p'$, and finally a path from $q'$ to $p$, not containing $p'$ and $q$. If we perform path-contractions along these paths, leaving only $p$, $q$, $p'$ and $p'$, the resulting graph is the $K_4$, which is of tree-width 3. We reach a contradiction and can conclude that $G$ must be plane.
\end{proof}

The desired lower bound for spanners of tree-width $2$ now follows from~\cite{Dumitrescu2016} combined with Lemma \ref{lem:tw2spannerplane}.
\begin{proposition}
Let $P$ be a set of 23 points equally spaced on a circle. Every spanner of tree-width 2 on $P$ has dilation at least 1.430814.   
\end{proposition}

\section{Conclusion and open questions}
In this paper, we give the first algorithm for computing geometric spanners with bounded tree-width for points in Euclidean space of constant dimension. Further, we show a strong dependency between the tree-width and the number of edges in connected planar graphs with few edges. 

Our lower bound proves that the dilation of our spanner is worst-case optimal for tree-width $k$. This naturally leads to the question of optimisation: Given a set of points and a parameter $k$, can we compute the minimum-dilation spanner of tree-width $k$? Already for $k=1$, that is for minimum-dilation trees, the problem is NP-hard \cite{klein2007computing, CHEONG2008188}. This raises the question of approximation:
Can we efficiently approximate the minimum-dilation spanner of tree-width $k$ in terms of its dilation within a factor better than $O(n/k^{d/d-1})$? As a bicriteria problem, the tree-width of the approximation may also be $f(k)$ for a suitable function $f$.

\bibliography{bib}

\newcommand{\SortNoop}[1]{}
\begin{thebibliography}{10}

\bibitem{Abramowitz1988}
Milton Abramowitz, Irene~A. Stegun, and Robert~H. Romer.
\newblock Handbook of mathematical functions with formulas, graphs, and
  mathematical tables.
\newblock {\em American Journal of Physics}, 56(10), 1988.

\bibitem{AgarwalES91}
Pankaj~K. Agarwal, Herbert Edelsbrunner, and Otfried Schwarzkopf.
\newblock Euclidean minimum spanning trees and bichromatic closest pairs.
\newblock {\em Discret. Comput. Geom.}, 6:407--422, 1991.
\newblock \href {https://doi.org/10.1007/BF02574698}
  {\path{doi:10.1007/BF02574698}}.

\bibitem{Althfer1993}
Ingo Althöfer, Gautam Das, David Dobkin, Deborah Joseph, and José Soares.
\newblock On sparse spanners of weighted graphs.
\newblock {\em Discrete \& Computational Geometry}, 9(1):81–100, 1993.
\newblock \href {https://doi.org/10.1007/BF02189308}
  {\path{doi:10.1007/BF02189308}}.

\bibitem{Aronov2008}
Boris Aronov, Mark {de Berg}, Otfried Cheong, Joachim Gudmundsson, Herman
  Haverkort, Michiel Smid, and Antoine Vigneron.
\newblock Sparse geometric graphs with small dilation.
\newblock {\em Computational Geometry}, 40(3):207--219, 2008.
\newblock \href {https://doi.org/10.1016/J.COMGEO.2007.07.004}
  {\path{doi:10.1016/J.COMGEO.2007.07.004}}.

\bibitem{Amani2016}
Ahmad Biniaz, Mahdi Amani, Anil Maheshwari, Michiel Smid, Prosenjit Bose, and
  Jean-Lou De~Carufel.
\newblock A plane 1.88-spanner for points in convex position.
\newblock {\em Journal of Computational Geometry}, page Vol. 7 No. 1, 2016.
\newblock \href {https://doi.org/10.20382/JOCG.V7I1A21}
  {\path{doi:10.20382/JOCG.V7I1A21}}.

\bibitem{bonichon2010}
Nicolas Bonichon, Cyril Gavoille, Nicolas Hanusse, and Ljubomir Perkovi{\'c}.
\newblock Plane spanners of maximum degree six.
\newblock In {\em Automata, Languages and Programming: 37th International
  Colloquium, ICALP 2010, Bordeaux, France, July 6-10, 2010, Proceedings, Part
  I 37}, pages 19--30. Springer, 2010.
\newblock \href {https://doi.org/10.1007/978-3-642-14165-2_3}
  {\path{doi:10.1007/978-3-642-14165-2_3}}.

\bibitem{bose2010computing}
Prosenjit Bose, Paz Carmi, Mohammad Farshi, Anil Maheshwari, and Michiel Smid.
\newblock Computing the greedy spanner in near-quadratic time.
\newblock {\em Algorithmica}, 58(3):711--729, 2010.
\newblock \href {https://doi.org/10.1007/s00453-009-9293-4}
  {\path{doi:10.1007/s00453-009-9293-4}}.

\bibitem{Bose.2013}
Prosenjit Bose and Michiel Smid.
\newblock On plane geometric spanners: A survey and open problems.
\newblock {\em Comput. Geom. Theory Appl.}, 46(7):818--830, 2013.
\newblock \href {https://doi.org/10.1016/j. comgeo.2013.04.002}
  {\path{doi:10.1016/j. comgeo.2013.04.002}}.

\bibitem{Cabello2022Computing}
Sergio Cabello.
\newblock Computing the inverse geodesic length in planar graphs and graphs of
  bounded treewidth.
\newblock {\em ACM Trans. Algorithms}, 18(2), mar 2022.
\newblock \href {https://doi.org/10.1145/3501303} {\path{doi:10.1145/3501303}}.

\bibitem{CABELLO2009815}
Sergio Cabello and Christian Knauer.
\newblock Algorithms for graphs of bounded treewidth via orthogonal range
  searching.
\newblock {\em Computational Geometry}, 42(9):815--824, 2009.
\newblock \href {https://doi.org/10.1016/j.comgeo.2009.02.001}
  {\path{doi:10.1016/j.comgeo.2009.02.001}}.

\bibitem{cabello2010obnoxious}
Sergio Cabello and G{\"u}nter Rote.
\newblock Obnoxious centers in graphs.
\newblock {\em SIAM Journal on Discrete Mathematics}, 24(4):1713--1730, 2010.
\newblock \href {https://doi.org/10.1137/09077638X}
  {\path{doi:10.1137/09077638X}}.

\bibitem{CHEONG2008188}
Otfried Cheong, Herman Haverkort, and Mira Lee.
\newblock Computing a minimum-dilation spanning tree is {NP}-hard.
\newblock {\em Computational Geometry}, 41(3):188--205, 2008.
\newblock \href {https://doi.org/10.1016/j.comgeo.2007.12.001}
  {\path{doi:10.1016/j.comgeo.2007.12.001}}.

\bibitem{chew1986}
Paul Chew.
\newblock There is a planar graph almost as good as the complete graph.
\newblock In {\em Proc.\ 2nd Symposium on Computational Geometry}, pages
  169--177, 1986.
\newblock \href {https://doi.org/10.1145/10515.10534}
  {\path{doi:10.1145/10515.10534}}.

\bibitem{CO00}
B.~Courcelle and S.~Olariu.
\newblock Upper bounds to the clique width of graphs.
\newblock {\em Discrete Applied Mathematics}, 101:77--114, 2000.
\newblock \href {https://doi.org/10.1016/s0166-218x(99)00184-5}
  {\path{doi:10.1016/s0166-218x(99)00184-5}}.

\bibitem{das1996}
Gautam Das and Paul~J Heffernan.
\newblock Constructing degree-3 spanners with other sparseness properties.
\newblock {\em International Journal of Foundations of Computer Science},
  7(02):121--135, 1996.
\newblock \href {https://doi.org/10.1007/3-540-57568-5_230}
  {\path{doi:10.1007/3-540-57568-5_230}}.

\bibitem{DasJ89}
Gautam Das and Deborah Joseph.
\newblock Which triangulations approximate the complete graph?
\newblock In {\em Optimal Algorithms}, volume 401 of {\em LNCS}, pages
  168--192. Springer-Verlag, 1989.
\newblock \href {https://doi.org/10.1007/3-540-51859-2_15}
  {\path{doi:10.1007/3-540-51859-2_15}}.

\bibitem{Dinneen1995}
Michael~John Dinneen.
\newblock {\em Bounded Combinatorial Width and Forbidden Substructures}.
\newblock PhD thesis, University of Victoria, 1995.

\bibitem{drysdale2001exclusion}
RL~Scot Drysdale, Scott McElfresh, and Jack~Scott Snoeyink.
\newblock On exclusion regions for optimal triangulations.
\newblock {\em Discrete Applied Mathematics}, 109(1-2):49--65, 2001.
\newblock \href {https://doi.org/10.1016/S0166-218X(00)00236-5}
  {\path{doi:10.1016/S0166-218X(00)00236-5}}.

\bibitem{Dumitrescu2016}
Adrian Dumitrescu and Anirban Ghosh.
\newblock Lower bounds on the dilation of plane spanners.
\newblock {\em International Journal of Computational Geometry \&
  Applications}, 26(02):89–110, 2016.
\newblock \href {https://doi.org/10.1142/S0218195916500059}
  {\path{doi:10.1142/S0218195916500059}}.

\bibitem{Dvok2019}
Zdeněk Dvořák and Sergey Norin.
\newblock Treewidth of graphs with balanced separations.
\newblock {\em Journal of Combinatorial Theory, Series B}, 137:137–144, 2019.
\newblock \href {https://doi.org/10.1016/J.JCTB.2018.12.007}
  {\path{doi:10.1016/J.JCTB.2018.12.007}}.

\bibitem{Eppstein00}
David Eppstein.
\newblock Spanning trees and spanners.
\newblock In {\em Handbook of Computational Geometry}, pages 425--461.
  Elsevier, 2000.
\newblock \href {https://doi.org/h10.1016/B978-044482537-7/50010-3}
  {\path{doi:h10.1016/B978-044482537-7/50010-3}}.

\bibitem{EppsteinK21}
David Eppstein and Hadi Khodabandeh.
\newblock On the edge crossings of the greedy spanner.
\newblock In Kevin Buchin and {\'{E}}ric~Colin de~Verdi{\`{e}}re, editors, {\em
  37th International Symposium on Computational Geometry, SoCG 2021}, volume
  189 of {\em LIPIcs}, pages 33:1--33:17. Schloss Dagstuhl - Leibniz-Zentrum
  f{\"{u}}r Informatik, 2021.
\newblock URL: \url{https://doi.org/10.4230/LIPIcs.SoCG.2021.33}, \href
  {https://doi.org/10.4230/LIPICS.SOCG.2021.33}
  {\path{doi:10.4230/LIPICS.SOCG.2021.33}}.

\bibitem{Grohe2009}
Martin Grohe and Dániel Marx.
\newblock On tree width, bramble size, and expansion.
\newblock {\em Journal of Combinatorial Theory, Series B}, 99(1):218–228,
  2009.
\newblock \href {https://doi.org/10.1016/j.jctb.2008.06.004}
  {\path{doi:10.1016/j.jctb.2008.06.004}}.

\bibitem{gudmundsson2018dilation}
Joachim Gudmundsson and Christian Knauer.
\newblock Dilation and detours in geometric networks.
\newblock In {\em Handbook of Approximation Algorithms and Metaheuristics},
  pages 53--69. Chapman and Hall/CRC, 2018.

\bibitem{Jeffrey2008}
Alan Jeffrey and Hui-hui Dai.
\newblock {\em Handbook of Computational Geometry}.
\newblock Academic Press, 2008.

\bibitem{kanj.2017}
Iyad Kanj, Ljubomir Perkovic, and Duru Turkoglu.
\newblock Degree four plane spanners: Simpler and better.
\newblock {\em J.\ Comput. Geom.}, 8(2):3--31, 2017.
\newblock \href {https://doi.org/10.20382/jocg.v8i2a2}
  {\path{doi:10.20382/jocg.v8i2a2}}.

\bibitem{klein2007computing}
Rolf Klein and Martin Kutz.
\newblock Computing geometric minimum-dilation graphs is {NP}-hard.
\newblock In {\em Graph Drawing: 14th International Symposium, GD 2006,
  Karlsruhe, Germany, September 18-20, 2006. Revised Papers 14}, pages
  196--207. Springer, 2007.
\newblock \href {https://doi.org/10.1007/978-3-540-70904-6_20}
  {\path{doi:10.1007/978-3-540-70904-6_20}}.

\bibitem{Le2024}
Hung Le and Cuong Than.
\newblock Greedy spanners in euclidean spaces admit sublinear separators.
\newblock {\em ACM Transactions on Algorithms}, 20(3):1–30, 2024.
\newblock \href {https://doi.org/10.1145/3590771} {\path{doi:10.1145/3590771}}.

\bibitem{Lipton1979}
Richard~J. Lipton and Robert~Endre Tarjan.
\newblock A separator theorem for planar graphs.
\newblock {\em SIAM Journal on Applied Mathematics}, 36(2):177–189, 1979.
\newblock \href {https://doi.org/10.1137/0136016} {\path{doi:10.1137/0136016}}.

\bibitem{Narasimhan.2007}
Giri Narasimhan and Michiel Smid.
\newblock {\em Geometric spanner networks}.
\newblock Cambridge University Press, 2007.
\newblock \href {https://doi.org/10.1017/cbo9780511546884}
  {\path{doi:10.1017/cbo9780511546884}}.

\bibitem{Robertson1994}
N.~Robertson, P.~Seymour, and R.~Thomas.
\newblock Quickly excluding a planar graph.
\newblock {\em Journal of Combinatorial Theory, Series B}, 62(2):323–348,
  1994.
\newblock \href {https://doi.org/10.1006/JCTB.1994.1073}
  {\path{doi:10.1006/JCTB.1994.1073}}.

\bibitem{RS83}
N.~Robertson and P.D. Seymour.
\newblock Graph minors {I}. {E}xcluding a forest.
\newblock {\em Journal of Combinatorial Theory, Series B}, 35:39--61, 1983.
\newblock \href {https://doi.org/10.1016/0095-8956(83)90079-5}
  {\path{doi:10.1016/0095-8956(83)90079-5}}.

\bibitem{RS86}
N.~Robertson and P.D. Seymour.
\newblock Graph minors {II}. {A}lgorithmic aspects of tree width.
\newblock {\em Journal of Algorithms}, 7:309--322, 1986.
\newblock \href {https://doi.org/10.1016/0196-6774(86)90023-4}
  {\path{doi:10.1016/0196-6774(86)90023-4}}.

\bibitem{RobinsS95}
Gabriel Robins and Jeffrey~S. Salowe.
\newblock Low-degree minimum spanning trees.
\newblock {\em Discret. Comput. Geom.}, 14(2):151--165, 1995.
\newblock \href {https://doi.org/10.1007/BF02570700}
  {\path{doi:10.1007/BF02570700}}.

\bibitem{Sattari2019}
Sattar Sattari and Mohammad Izadi.
\newblock An improved upper bound on dilation of regular polygons.
\newblock {\em Computational Geometry}, 80:53–68, 2019.
\newblock \href {https://doi.org/10.1016/J.COMGEO.2019.01.009}
  {\path{doi:10.1016/J.COMGEO.2019.01.009}}.

\bibitem{Seidman1983}
Stephen~B. Seidman.
\newblock Network structure and minimum degree.
\newblock {\em Social Networks}, 5(3):269–287, 1983.

\bibitem{smid2007well}
Michiel~HM Smid.
\newblock The well-separated pair decomposition and its applications.
\newblock {\em Handbook of approximation algorithms and metaheuristics}, 13,
  2007.
\newblock \href {https://doi.org/10.1201/9781420010749.ch53}
  {\path{doi:10.1201/9781420010749.ch53}}.

\bibitem{xia2013}
Ge~Xia.
\newblock The stretch factor of the delaunay triangulation is less than 1.998.
\newblock {\em SIAM Journal on Computing}, 42(4):1620--1659, 2013.
\newblock \href {https://doi.org/10.1137/110832458}
  {\path{doi:10.1137/110832458}}.

\end{thebibliography}

\end{document}